\documentclass{ifacconf}

\usepackage{natbib}        
\usepackage{graphicx,nicefrac}
\usepackage{latexsym,amsfonts,amssymb,amsmath,amscd,mathrsfs,mathtools}
\usepackage{xspace}
\usepackage{float}
\usepackage{algorithm}
\usepackage{caption}
\usepackage{enumitem}

\usepackage[left,mathlines]{lineno}
\newtheorem{theorem}{Theorem}

\newtheorem{definition}{Definition}
\newtheorem{proposition}[theorem]{Proposition}

\newtheorem{remark}{Remark}
\newtheorem{example}{Example}

\newcommand{\inte}{\text{int}}

\newcommand{\R}{\ensuremath{\mathbb{R}}\xspace} 
\newcommand{\N}{\ensuremath{\mathbb{N}}\xspace} 
 
\newcommand{\W}{\ensuremath{\mathbb{W}}\xspace} 
\newcommand{\D}{\ensuremath{\mathbb{D}}\xspace} 

\newcommand{\X}{\ensuremath{\mathbb{X}}\xspace}

\newenvironment{proof}[1][Proof]{\par\noindent\textit{#1.} }{\hfill$\square$\par}

\newcounter{myalg}

\newenvironment{myalgorithm}[1][]{
    \refstepcounter{myalg}
    \par\medskip
    \noindent\textbf{Algorithm~\themyalg. #1 Set attractivity for uncertain linear switched systems~\eqref{eq:linearss}.}
    \begin{enumerate}
}{
    \end{enumerate}\par\medskip
}

\begin{document}

\begin{frontmatter}

\title{Certifying Set Attractivity for Discrete-Time Uncertain Nonlinear Switched Systems\thanksref{footnoteinfo}} 

\thanks[footnoteinfo]{Work funded by the European Union through the ERC INSPIRE grant (project number 101076926); views and opinions expressed are however those of the authors only and do not necessarily reflect those of the European Union or the European Research Council; neither the European Union nor the granting authority can be held responsible for them.}

\author[First]{Alejandro Anderson} 
\author[Second]{Esteban A. Hernandez-Vargas} 
\author[First]{Giulia Giordano}

\address[First]{Dipartimento di Ingegneria Industriale, Università di Trento, Trento, Italy (e-mail: alelanderson@gmail.com, giulia.giordano@unitn.it).}
\address[Second]{Department of Mathematics and Statistical Science, University of Idaho, Moscow, Idaho, USA (e-mail: esteban@uidaho.edu)}

\begin{abstract}
We introduce a new class of functions, called \emph{Attractivity Guarantee (AG)-functions}, to certify the attractivity of sets for uncertain nonlinear switched systems in discrete time. The existence of an AG-function associated with a set guarantees the robust local attractivity of that set under the system dynamics.
We propose a constructive method for obtaining piecewise-continuous AG-functions based on contractive sets for the system, 
and show that the existence of a robust control contractive set for the dynamics implies the existence of an appropriate AG-function, and hence the robust local attractivity of the set itself.
We illustrate the proposed framework through examples that elucidate the theoretical concepts, and through the case study of a nonlinear switched system modelling antimicrobial resistance, which highlights the practical relevance of the approach to the analysis of biological systems.
\end{abstract}

\end{frontmatter}
%===============================================================================

\section{Introduction}
Switched systems model complex systems where discrete factors, or switching laws, influence the continuous evolution of the states \citep{liberzon1999basic,liberzon2003switching}, making the system dynamics switch between different configurations, called modes.
When the switching law is a decision variable, switched systems provide a natural framework for practical biomedical applications \citep{anderson2021}, such as multi-drug cycling strategies in the treatment of cancer \citep{giordano2016,devia2019,wu2022switched} and of viral and bacterial infections \citep{kouyos2011informed,hernandez2011discrete,hernandez2012,hernandez2013,hernandez2014,hernandez2021switching,Tetteh2023,anderson2024stabilizability,anderson2025computational,anderson2025invariant}, functional electrical stimulation cycling for patients with neuromuscular disorders \citep{cousin2021switched} and insulin sensitivity management in  artificial pancreas \citep{cavallo2025insulin}. Other application areas include ecological systems \citep{anderson2023cyclic}, as well as power systems \citep{russo2022state}, traffic control \citep{hajiahmadi2016robust} and robotic systems \citep{aguiar2005stability,lee2008uniform,chen2019switched}.

A vast body of literature has investigated the stability analysis of switched systems, as well as the design of stabilising switching laws \citep{liberzon1999basic,liberzon2003switching,sun2011stability}, with a particular focus on linear and affine switched systems \citep{lin2009stability,jungers2017feedback,heemels2016lyapunov}, both in continuous time \citep{hespanha2004,geromel2006stabilityCT,wang2019stability} and in discrete time \citep{geromel2006stability,fiacchini2014necessary,deaecto2018stability,fiacchini2018stabilization,gomide2018stability,egidio2019,fiacchini2021yet}. Nonlinear switched systems have also been considered, both in continuous time \citep{colaneri2008stabilization,lee2008uniform,zhao2008stability,yang2014survey,liu2019separable,khoa2022exponential,zagabe2025uniform} and in discrete time \citep{liu2015stability,zhang2020stability,deaecto2025stabilisation}.

Although the literature is mainly devoted to uncertainty-free switched systems, the robust stability analysis and stabilisation of uncertain switched systems has also been investigated, particularly in the linear case \citep{lin2007switching,allerhand2010robust,son2020robust}, and in the nonlinear case too \citep{aguiar2005stability,niu2013robust,sun2013stability,hajiahmadi2016robust,noghredani2021robust}.

A powerful approach for the stability analysis of linear switched systems is offered by Lyapunov methods, which also enable the design of a switching law that stabilises the system  \citep{geromel2006stability,lin2009stability,deaecto2018stability}; yet, finding a common Lyapunov function that ensures stability is a nontrivial task \citep{mason2023universal}. Even more so, proving the existence of a stabilising switching law for nonlinear switched systems through a Lyapunov-based analysis is extremely challenging in general \citep{shorten2007stability}.

An alternative to classic Lyapunov analysis relies on controllability of invariant sets \citep{Blanchinibook15}, which provides a naturally robust framework to tackle uncertainties, identify stable and attractive regions, and offer geometric certificates of asymptotic stability and constraint satisfaction \citep{danielson2019necessary}. The concepts of control invariance and set attractivity have been extensively studied for linear switched systems \citep{xiang2017reachable,fiacchini2021yet,cinto2025switching,perez2025characterization}.
The set-theoretic approach by \cite{fiacchini2014necessary}, providing necessary and sufficient conditions for the stabilisability of uncertainty-free linear switched systems based on the existence of a contractive invariant set, was later extended by \cite{anderson2024stabilizability} to account for uncertainties.
However, attractivity analysis based on invariant sets remains
a challenging open problem for nonlinear switched systems subject to uncertainties; to the best of our knowledge, no general result is currently available for \textit{uncertain nonlinear switched systems}.

Here, we address the problem of certifying set attractivity for uncertain nonlinear switched systems with state constraints in discrete time, when the switching law is a decision variable. We introduce the class of \emph{AG-functions} — a relaxation of the concept of control Lyapunov function, with weaker requirements — and show that the existence of an AG-function associated with a set guarantees robust local attractivity of that set for the dynamics.
We also propose a particular construction of AG-functions based on robust contractive invariant sets for the system, and we prove that the existence of a contractive set implies the existence of an AG-function, and hence robust local attractivity.
Our main results are visualised in Fig.~\ref{fig:diagram}, which illustrates their relationships and summarises their logical structure.
We demonstrate our theoretical results, in a practical application domain, on a switched system that models bacterial population dynamics in a host, affected by the immune response and antibiotic action, in the presence of antimicrobial resistance (AMR); numerical simulations reveal the presence of invariant contractive sets, and hence of attractive sets for the system dynamics. This result is relevant for assessing the fundamental limitations of antibiotic therapy protocols in preventing AMR, depending on the initial bacterial load.

\subsection{Notation}
We denote the sets $\mathbb{N}=\{0,1,2,\dots\}$ and $\mathbb{N}_q:=\{1,\dots,q\}$. Given $k\ge 1$ and a set $\Omega\subseteq\R^n$, we define $\Omega^k := \Omega\times \dots \times \Omega$, $k$ times, where $\times$ denotes the Cartesian product.
We define $d(x,y):=\|x-y\|=[(x-y)^\top(x-y)]^{1/2}$ the Euclidean distance between $x,y \in \R^n$;
the point-to-set distance between $x\in\mathbb{R}^n$ and a closed set $\Omega\subset\mathbb{R}^n$ is $d(x, \Omega) := \min_{y \in \Omega} \|x - y\|$.
Given the compact ambient set $\X\subseteq\R^n$, the open ball with center in $x\in\X$ and radius $\varepsilon>0$ relative to $\X$ is $\mathcal{B}_{\X}(x,\varepsilon):=\{y\in\X: d(x,y)<\varepsilon\}$, while the closed ball is $\bar{\mathcal{B}}_{\X}(x,\varepsilon):=\{y\in\X: d(x,y)\leq \varepsilon\}$.
Given $\Omega \subset \R^n$, the \emph{interior of $\Omega$ relative to $\X$} is the set $\inte_\X \left(\Omega\right) := \{x \in \Omega \cap \X \colon \exists \varepsilon > 0 \mbox{ such that } \mathcal{B}_{\X}(x,\varepsilon) \subseteq \Omega\}$.
Denote by $\text{cl}(\Omega)$ the closure of $\Omega \subset \R^n$ in $\R^n$. The \textit{boundary $\partial_\X \Omega$ of $\Omega$ relative to $\X$} consists of all points $x\in \text{cl}(\Omega) \cap \X$ such that, for all $\varepsilon$, $\mathcal{B}_{\X}(x,\varepsilon)$ contains at least one point in $\Omega$ and at least one point not in $\Omega$. Throughout the work, we consider interior and boundary of sets relative to $\X$.

%%%%%%%%%%%%%%%%%%%%%%%%%%%%%%%%%%%%%%%%%%%%%%%%%%%%%%%%%%%%%%%%%%%%%%%%%%%%%%%%%%%%%%%%%%%%%%
\section{Model Class and Preliminaries}\label{sec:Preliminaries}
%%%%%%%%%%%%%%%%%%%%%%%%%%%%%%%%%%%%%%%%%%%%%%%%%%%%%%%%%%%%%%%%%%%%%%%%%%%%%%%%%%%%%%%%%%%%%%
The dynamics of an uncertain  nonlinear switched system in discrete time can be described by the equation
\begin{align}\label{eq:SistOrig}
x(k+1)=f_{\sigma(k)}(x(k),w(k)),
\end{align}
where, at a given discrete time $k\ge 0$, $x(k)\in\X \subset \R^n$ is the state of the system, $w(k)\in\W \subset \R^n$ is the uncertainty, and $\sigma(k) \in \mathbb{N}_q$ is the active mode, selected according to the switching law $\sigma \colon \N\rightarrow \mathbb{N}_q$, a function that, at each time $k$, selects the active mode among $q>1$ possible modes.
The sets $\X$ and $\W$ are compact.
Function $f_{\sigma} \colon \X\times\W \to \X$ is jointly continuous in $\X\times\W$ for all $\sigma\in\mathbb{N}_q$.

For an initial state $x$, switching sequence (i.e., admissible sequence of active modes) $\sigma^k := \{\sigma(1),\dots,\sigma(k)\}\in\N_q^k$ and uncertainty sequence $w^k:=\{w(1),\dots,w(k)\}\in\W^k$ (i.e., $\sigma(i)\in\N_q$ and $w(i)\in\W$ for all $i=1,\dots,k$), with $k\geq 1$, we denote by $\phi(x,\sigma^k,w^k)$ the resulting state at time $k$, while $\boldsymbol{\Phi}(x, \sigma^k)$ denotes the set of all possible states reached at time $k$, starting from $x$, by applying the switching sequence $\sigma^k$, for all possible uncertainty sequences $w^k\in\W^k$; see the illustration in Fig.~\ref{fig:bundles}-A. The set $\boldsymbol{\Phi}(x,\sigma^k)=\{\phi(x,\sigma^k,w^k) \text{ for some } w^k\in\W^k\}$ does not depend on the specific uncertainty realisation, but only on the uncertainty set $\W$.
\begin{figure}[ht]
	\centering	\includegraphics[width=0.85\linewidth]{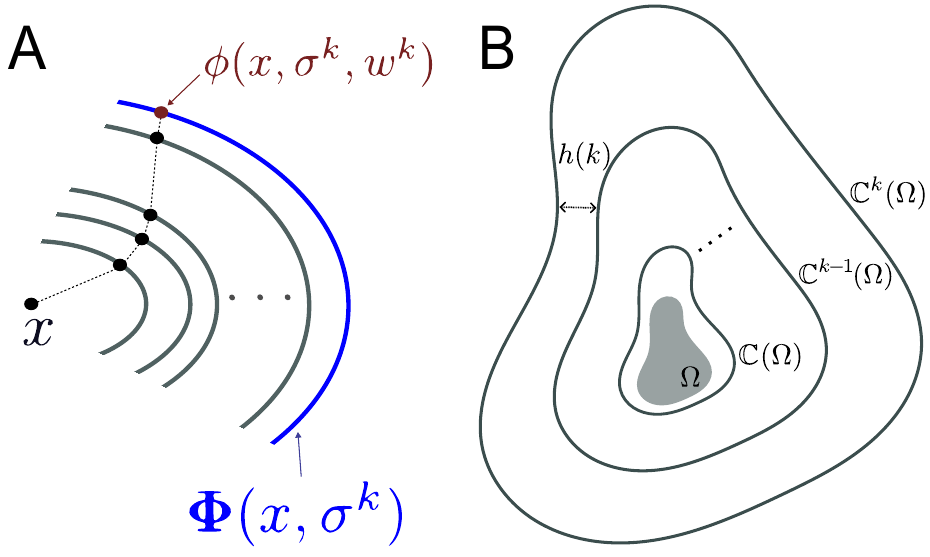}
	\caption{\small \textbf{A.} $\phi(x,\sigma^k,w^k)$ is the state reached at time $k$ from the initial state $x\in\X$, if the switching sequence $\sigma^k=\{\sigma(1),\dots,\sigma(k)\}$ is applied and the realised uncertainty sequence is $w^k=\{w(1),\dots,w(k)\}$, while $\boldsymbol{\Phi}(x, \sigma^k)$ is the set of all states that can be reached at time $k$ from $x$, if the switching sequence $\sigma^k$ is applied, for all possible uncertainty sequences in $\W^k$. \textbf{B.} The robust control contractive set $\Omega$ is contained in the interior of $\mathbb{C}(\Omega)$; for all $k \ge 0$, the controllable sets satisfy $\mathbb{C}^{k}(\Omega) \subseteq \inte_\X \big(\mathbb{C}^{k+1}(\Omega)\big)$, while function $h(k)$ quantifies the distance between the boundaries of $\mathbb{C}^{k-1}(\Omega)$ and of $\mathbb{C}^{k}(\Omega)$.}
	\label{fig:bundles}
\end{figure}

We provide preliminary definitions from the set-theoretic control literature (see, e.g., \citealp{Blanchinibook15}).

We wish to analyse whether a given \textit{robust control invariant set} is attractive under the system dynamics~\eqref{eq:SistOrig}.
\begin{definition}\textbf{(Robust control invariant set, RCIS.)}\label{def:rcis}
Set $\Omega\subset\X$ is a \textit{robust control invariant set} for system~\eqref{eq:SistOrig} if, for every state $x\in\Omega$, there is $\sigma^1\in\N_q$ such that $\boldsymbol{\Phi}(x,\sigma^1)\subseteq\Omega$.
\end{definition}

The $1$-step condition in Definition~\ref{def:rcis} implies, for all $k \ge 1$, that a switching sequence $\sigma^k \in \N_q^k$ can be recursively constructed such that $\boldsymbol{\Phi}(x,\sigma^k) \subseteq \Omega$; the inclusion can be thus guaranteed to hold for all future steps.
Henceforth, whenever we mention a RCIS, we implicitly assume that it is nonempty and compact.
A RCIS for the system is locally attractive if it has a nontrivial \textit{domain of attraction}.

\begin{definition}\textbf{(Domain of attraction.)}  \label{def:domain}
The \textit{domain of attraction} $\D(\Omega)$ of a RCIS $\Omega\subset\X$ for system~\eqref{eq:SistOrig} is given by all states $x\in\X$ starting from which the system trajectory asymptotically reaches $\Omega$ ($\lim_{k\to \infty} d(\phi(x,\sigma^{k},w^{k}),\Omega)=0$) for some switching sequence $\sigma^{k}\in\N_q^{k}$ and for all admissible uncertainty sequences $w^{k}\in\W^{k}$.
\end{definition}
Consider the \textit{controllable set} (also called \emph{backward reachable set}) of states from which the dynamics of system~\eqref{eq:SistOrig} can be driven into a given RCIS \citep{kerrigan2000invariant,chen2016estimation,chen2018reachable,baldi2018reachable,decardi2021computing,serry2023zonotopic}.

\begin{definition}\textbf{(Robust controllable set for mode $\sigma^1$.)} \label{def:1-scssigma}
Given a RCIS $\Omega \subset \X$ for system~\eqref{eq:SistOrig}, the \textit{robust controllable set} to $\Omega$ for mode $\sigma^1\in\N_q$ is
\begin{equation}
 \mathbb{C}(\sigma^1,\Omega) = \{ x \in \X \colon  \boldsymbol{\Phi}(x, \sigma^1) \subseteq \Omega \}.
\end{equation}
\end{definition}

\begin{definition}\textbf{(Robust $1$-step controllable set.)}\label{def:1-scs}
Given a RCIS $\Omega \subset \X$ for system~\eqref{eq:SistOrig}, the \textit{robust $1$-step controllable set} to $\Omega$ is
\begin{equation}
\mathbb{C}(\Omega) = \bigcup_{\sigma^1\in\N_q} \mathbb{C}(\sigma^1,\Omega),
\end{equation}
i.e., $\mathbb{C}(\Omega) = \{x \in \X \colon \boldsymbol{\Phi}(x, \sigma^1)\subseteq\Omega \text{ for some } \sigma^1\in\N_q\}$.
\end{definition}

\begin{remark}
To verify that a set $\Omega$ is a RCIS, it is sufficient to check that $\Omega \subseteq \mathbb{C}(\Omega)$.
\end{remark}

\begin{definition}\textbf{(Robust $k$-step controllable set.)} \label{def:k-scs}
Given a RCIS $\Omega \subset \X$ for system~\eqref{eq:SistOrig}, the \textit{robust $k$-step controllable set} to $\Omega$ for system~\eqref{eq:SistOrig} can be defined recursively by 
\[\mathbb{C}^k(\Omega) := \mathbb{C}(\mathbb{C}^{k-1}(\Omega)),\] 
for $k\ge 1$, with $\mathbb{C}^1(\Omega):= \mathbb{C}(\Omega)$ and $\mathbb{C}^0(\Omega):=\Omega$.
Equivalently, $\mathbb{C}^k(\Omega) = \{x \in \X \colon \boldsymbol{\Phi}(x, \sigma^{k})\subseteq\Omega \text{ for some } \sigma^k\in\N_q^k\}$.
\end{definition}

Given a RCIS $\Omega \subset \X$ for system~\eqref{eq:SistOrig}, we have
\begin{equation}\label{eq:D_Omega}
\mathbb{C}_\infty(\Omega):=\bigcup_{k=0}^\infty \mathbb{C}^k(\Omega)\subseteq \D(\Omega).
 \end{equation}

As a standing assumption throughout the work, we let $\Omega \subset \X$ be a closed (and hence compact) set.

\begin{remark}
Since $\Omega$ is closed, for each $k \geq 1$, the set $\mathbb{C}^{k}(\Omega)$ is a finite union of robust controllable sets, each associated with a mode. Since each of these sets is closed \citep{blanchini1994ultimate}, their finite union is also closed. In contrast, the set $\mathbb{C}_\infty(\Omega)$ is defined as an infinite union of closed sets, which does not necessarily preserve closedness. 
Hence, while every $\mathbb{C}^{k}(\Omega)$ is closed, $\mathbb{C}_\infty(\Omega)$ may fail to be closed.
\end{remark}

Robust controllable sets to a RCIS have a nested structure.
\begin{proposition}\label{proper:nestedC}
Given a RCIS $\Omega \subset \X$ for system~\eqref{eq:SistOrig},
\[ \mathbb{C}^k(\Omega) \subseteq \mathbb{C}^{k+1}(\Omega) \,\, \mbox{ for all } k\ge 0.\]
\end{proposition}
\begin{proof}
Consider $k=0$. Since $\Omega$ is a RCIS, by Definition~\ref{def:rcis}, for all $x\in\Omega$ there exists $\sigma^1 \in\N_q$ such that $\boldsymbol{\Phi}(x, \sigma^1)\subseteq\Omega$, and hence $x\in\mathbb{C}(\Omega)$ as per Definition~\ref{def:1-scs}. 
Consider $k\ge 1$. For all $x\in\mathbb{C}^{k}(\Omega)$, by Definition~\ref{def:k-scs} there exists a switching sequence $\sigma^{k}=\{\sigma(1), \dots, \sigma(k)\}$ such that $\boldsymbol{\Phi}(x, \sigma^{k})\subseteq\Omega$. Since $\Omega$ is a RCIS, there must exist $\sigma(k+1)\in\N_q$ such that $\boldsymbol{\Phi}(x, \sigma^{k+1})\subseteq\Omega$ for $\sigma^{k+1}:=\{\sigma(1), \cdots, \sigma(k),\sigma(k+1)\}$, which implies $x\in\mathbb{C}^{k+1}(\Omega)$.
\end{proof}

%%%%%%%%%%%%%%%%%%%%%%%%%%%%%%%%%%%%%%%%%%%%%%%%%%%%%%%%%%%%%%%%%%%%%%%%%%%%%%%%%%%%%%%%%%%%%%%%%%%%%%%%%
\section{AG-functions for set attractivity}\label{sec:TheLFunction}
%%%%%%%%%%%%%%%%%%%%%%%%%%%%%%%%%%%%%%%%%%%%%%%%%%%%%%%%%%%%%%%%%%%%%%%%%%%%%%%%%%%%%%%%%%%%%%%%%%%%%%%%%

We can associate an AG-function with a RCIS $\Omega$ for system~\eqref{eq:SistOrig}, as per the following definition.
\begin{definition}\textbf{(AG-function})\label{def:L-function}
Given a RCIS $\Omega\subset\X$ for system~\eqref{eq:SistOrig}, function $L \colon \X\to \R_{\ge 0}$ is an \textit{AG-function} for system~\eqref{eq:SistOrig} and set $\Omega$ if it satisfies the following properties:\\
\textbf{(i)} $L(x)=0$ for all $x\in\Omega$ and $L(x)>0$ for all $x\in \X \setminus\Omega$.\\
\textbf{(ii)} There exists $R>0$ such that $\{x \in \X \colon d(x,\Omega)= R\}$ is nonempty and, given $L_R^\ast := \inf_{d(x,\Omega)=R} L(x)$ and defined the level set $\mathcal{L}_R(\Omega) := \{x\in\X \colon L(x)\le L_R^\ast\}$, there exists a non-decreasing function $\eta_R \colon [0,R]\to\R_{\ge 0}$, with $\eta_R(0)=0$ and $\eta_R(r)>0$ for all $r \in (0,R]$, such that
\[
\min_{\sigma^1\in\mathbb{N}_q}\;\sup_{z\in \boldsymbol{\Phi}(x,\sigma^1)}  L(z)-L(x)
\;\le\; -\eta_R\!\left(d(x,\Omega)\right)
\]
for all $x\in \mathcal{L}_R(\Omega)$.
\end{definition}

\begin{remark}
The control-Lyapunov-like AG-function in Definition~\ref{def:L-function} may be non-continuous; discontinuities may arise, e.g., due to switching between distinct modes of system~\eqref{eq:SistOrig}. Still, function $L$ is monotonically non-increasing along the system trajectories emanating from $\mathcal{L}_R(\Omega)$, and monotonically decreasing when $x \in \mathcal{L}_R(\Omega) \setminus \Omega$, as we show next.
\end{remark}

\begin{proposition}\label{prop:L-function-implies-attractivity}
Given a RCIS $\Omega\subset\X$ for system~\eqref{eq:SistOrig}, if there exists an AG-function for system~\eqref{eq:SistOrig} and set $\Omega$, as per Definition~\ref{def:L-function}, then $\Omega$ is robustly locally attractive in $\mathcal{L}_R(\Omega)$.
\end{proposition}

\begin{proof}
Consider a suitable $R>0$ and functions $\eta_R$ and $L$ as in Definition~\ref{def:L-function}.
The switching law can be always chosen so that $L$ is monotonically non-increasing along the system trajectories.
In fact, given a generic $k \geq 0$, for every admissible uncertainty sequence $w^{k+1}\in\W^{k+1}$, there exists an admissible switching sequence $\sigma^{k+1}\in\N_q^{k+1}$ such that, for the resulting trajectory of system~\eqref{eq:SistOrig}, 
with $\phi(x,\sigma^{k},w^{k})\in \mathcal{L}_R(\Omega)$,
property \textbf{(ii)} in Definition~\ref{def:L-function} yields
\begin{align}
L(\phi(x,\sigma^{k+1},w^{k+1})) - L(\phi(x,\sigma^{k},w^{k})) \nonumber\\
\le\ -\eta_R\!\bigl(d(\phi(x,\sigma^{k},w^{k}),\Omega)\bigr),
\label{eq:Ldec}
\end{align}
where $\phi(x,\sigma^{0},w^{0}):=x$.

Property \textbf{(ii)} in Definition~\ref{def:L-function} involves the supremum of $L(z)-L(x)$ over all $z\in \boldsymbol{\Phi}(x,\sigma^1)$.  
Therefore, possible discontinuities of $L$ do not invalidate the decrease condition~\eqref{eq:Ldec}.  
For any $x\in\mathcal{L}_R(\Omega)$, property \textbf{(ii)} implies that, $\forall w^1 \in \W$, there exists $\sigma^1 \in \mathbb{N}_q$
such that $L(\phi(x,\sigma^1,w^1))\le L(x)\le L_R^\ast$,
and thus $\phi(x,\sigma^1,w^1)\in\mathcal{L}_R(\Omega)$. Hence, $\mathcal{L}_R(\Omega)$ is a RCIS.
Consequently, given $x \in \mathcal{L}_R(\Omega)$, the decrease inequality~\eqref{eq:Ldec} is valid for all~$k \ge 0$.
Hence, the sequence $\{L(\phi(x,\sigma^{k},w^{k}))\}_{k \ge 0}$ is non-increasing within $\mathcal{L}_R(\Omega)$, and strictly decreasing within $\mathcal{L}_R(\Omega) \setminus \Omega$, and it is bounded below by~$0$.

Summing~\eqref{eq:Ldec} from $k=0$ to $N-1$,
and using the fact that $L$ takes non-negative values, yields
    \begin{align*}
    \sum_{k=0}^{N-1}\eta_R\!\big(d(\phi(x,\sigma^{k},w^{k}),\Omega)\big)\ &\le\ L(x)-L(\phi(x,\sigma^{N},w^{N})) \\ &\le\ L(x),
    \end{align*}
which implies that
\begin{equation}\label{eq:eta_R_proof}
\sum_{k=0}^{\infty}\eta_R(d(\phi(x,\sigma^{k},w^{k}),\Omega))<\infty,
\end{equation}
and hence 
$\lim_{k\to\infty} \eta_R\!\big(d(\phi(x,\sigma^{k},w^{k}),\Omega)\big) = 0$.
Since $\eta_R$ is non-decreasing and $\eta_R(r)>0$ for all $R \ge r>0$, while $\eta_R(0)=0$, this implies $\lim_{k\to\infty} d(\phi(x,\sigma^{k},w^{k}),\Omega)= 0$, as can be proven by contradiction.

In fact, suppose that $\eta_R\Big(d(\phi(x,\sigma^{k},w^{k}),\Omega)\Big) \to 0$, but the sequence $d(\phi(x,\sigma^{k},w^{k}),\Omega) \not\to 0$. Then, there must exist a subsequence such that $|d(\phi(x,\sigma^{k_j},w^{k_j}),\Omega)| > \alpha >0$ for all $j$. Then, since function $\eta_R$ is non-decreasing, we must also have $\eta_R\Big(d(\phi(x,\sigma^{k_j},w^{k_j}),\Omega)\Big) \geq \eta_R(\alpha)>0$ for all $j$, which however contradicts \eqref{eq:eta_R_proof}.  Note that continuity of $\eta_R$ is not required.

The above reasoning holds for all admissible uncertainty sequences and for all initial conditions $x$ in a neighbourhood of~$\Omega$.
Therefore, $\Omega$ is \emph{robustly locally attractive}, for all initial conditions in $\mathcal{L}_R(\Omega)$.
\end{proof}

%%%%%%%%%%%%%%%%%%%%%%%%%%%%%%%%%%%%%%%%%%%%%%%%%%%%%%%%%%%%%%%%%%%%%%%%%%%%%%%%%%%%%%%%%%%%%%%%%%%%%%%%%
\subsection{Construction of an AG-function}\label{sec:Construction_L_function}
%%%%%%%%%%%%%%%%%%%%%%%%%%%%%%%%%%%%%%%%%%%%%%%%%%%%%%%%%%%%%%%%%%%%%%%%%%%%%%%%%%%%%%%%%%%%%%%%%%%%%%%%%

Here, we provide a constructive method to generate AG-functions based on \textit{contractive sets} for the system.
\begin{definition}\textbf{(Robust control contractive set, RCCS.)}\label{def:RCCS}
Set $\Omega\subset\X$ is a \textit{robust control contractive set} for system~\eqref{eq:SistOrig} if, for every state $x\in\Omega$, there is $\sigma^1\in\N_q$ such that $\boldsymbol{\Phi}(x, \sigma^1)\subseteq \inte_\X\left(\Omega\right)$.
\end{definition}
\begin{remark}
The inclusion $\boldsymbol{\Phi}(x,\sigma^1)\subseteq \inte_\X(\Omega)$ in Definition~\ref{def:RCCS} differs from the classical condition $\boldsymbol{\Phi}(x,\sigma^1)\subseteq \lambda\Omega$, for some $\lambda\in(0,1)$, which requires the origin to lie in $\Omega$ (see, e.g., \citealp{Blanchinibook15}). This interior-based formulation was used by~\cite{doan2013hierarchical}.
\end{remark}

Of course, every RCCS is also a RCIS. Whenever we mention a RCCS, we implicitly assume that it is nonempty and compact.

\begin{remark}\label{rem:check_RCCS}
    Given any set $\Omega$ with $\Omega\subseteq \inte_\X\big(\mathbb C(\Omega) \big)$, the set $\mathbb C(\Omega)$ is a RCCS. In fact, if $x \in \mathbb C(\Omega)$, by Definition~\ref{def:1-scs}, there exists $\sigma^1 \in \N_q$ such that $\boldsymbol{\Phi}(x, \sigma^1) \subseteq \Omega \subseteq \inte_\X\big(\mathbb C(\Omega) \big)$, which corresponds to Definition~\ref{def:RCCS}.
\end{remark}

Robust controllable sets to a RCCS have a stronger nested structure that extends Proposition~\ref{proper:nestedC} (see Fig.~\ref{fig:bundles}-B).

\begin{proposition}\label{propo:cont_nesetedC}
If $\Omega$ is a RCCS of system~\eqref{eq:SistOrig}, then 
\begin{equation*}
\mathbb{C}^{k}(\Omega) \subset \inte_\X\left(\mathbb{C}^{k+1}(\Omega)\right)  \,\, \mbox{ for all } k\ge 0.
\end{equation*} 
\end{proposition}
\begin{proof}
We prove that $\Omega \subset \inte_\X\left(\mathbb{C}(\Omega)\right)$, which means that $\mathbb{C}(\Omega)$ is a RCCS; then, the result follows recursively (by applying the same reasoning with $\mathbb{C}(\Omega)$ instead of $\Omega$, and so on, in view of the recursive definition of $\mathbb{C}^{k}(\Omega)$).
Given $x \in \Omega$, by definition of RCCS there exists $\sigma^1 \in \mathbb{N}_q$ such that $\boldsymbol{\Phi}(x, \sigma^1)\subseteq \inte_\X\left(\Omega\right)$.
Since $\W$ is compact and $f_{\sigma^1}(x,\cdot)$ is continuous, then $\boldsymbol{\Phi}(x,\sigma^1)$ is compact, and hence there exists $\varepsilon > 0$ such that
\[
\bigcup_{z \in \boldsymbol{\Phi}(x, \sigma^1)} \bar{\mathcal{B}}_{\X}(z,\varepsilon) \subseteq \Omega,
\]
for the closed ball $\bar{\mathcal{B}}_{\X}(z,\varepsilon)$.
By joint continuity of $f_{\sigma^1}$ and compactness of $\W$, $f_{\sigma^1}$ is uniformly continuous on the compact set $\bar{\mathcal{B}}_{\X}(x,r)\times\W$ for some small $r>0$, in view of the Heine-Cantor theorem. Hence, for the considered $\varepsilon$, there exists $\delta\in(0,r)$ such that, for all $w\in\W$,
\[ \| y-x\|<\delta ~\Rightarrow \|f_{\sigma^1}(y,w) - f_{\sigma^1}(x,w) \|<\varepsilon, \]
and therefore, for all $y \in \bar{\mathcal{B}}_{\X}(x,\delta)$, it holds that
\begin{equation*}
\boldsymbol{\Phi}(y, \sigma^1) \subseteq \bigcup_{z \in \boldsymbol{\Phi}(x, \sigma^1)} \bar{\mathcal{B}}_{\X}(z,\varepsilon)\subseteq \Omega,
\end{equation*}
which implies that $\bar{\mathcal{B}}_{\X}(x,\delta) \subset \mathbb{C}(\Omega)$ and, hence, that $x \in \inte_\X\left(\mathbb{C}(\Omega)\right)$, which concludes the proof.
\end{proof}

We define the following function, illustrated in Fig.~\ref{fig:bundles}-B.
\begin{definition}\textbf{($h-$function.)}\label{def:h_function}
Given a RCIS $\Omega\subset\X$ for system~\eqref{eq:SistOrig}, the associated $h-$function $h \colon \N \to \R_{\ge 0}$ is $h(0):=0$ and, for $k\ge 1$,
\begin{equation}\label{eq:hx}
h(k):=\min\{d(y,z) \colon y \in \partial_\X \mathbb{C}^{k-1}(\Omega), z \in \partial_\X \mathbb{C}^{k}(\Omega)\}.
\end{equation}
\end{definition}

If $\Omega$ is a RCCS, then $h(k)>0$ for all $k\ge 1$, since the robust controllable sets are strictly nested as per Proposition~\ref{propo:cont_nesetedC}. We can thus prove the following.

\begin{proposition}\label{propo:decreasingd}
If $\Omega$ is a closed RCCS and $h$ is the associated $h$-function as per Definition~\ref{def:h_function}, then, for all $k \ge 0$,
\begin{equation}\label{eq:decreasingd}
 \min_{y \in \partial_\X \mathbb{C}^{k}(\Omega)} d(y, \Omega) \leq \min_{y \in \partial_\X \mathbb{C}^{k+1}(\Omega)} d(y, \Omega) -h(k+1).
 \end{equation}
\end{proposition}

\begin{proof}    
Recall that $\Omega$ and $\mathbb{C}^{k}(\Omega)$, for all $k$, are compact sets. When $k=0$, $\mathbb{C}^{k}(\Omega) = \Omega$, $\mathbb{C}^{k+1}(\Omega)=\mathbb{C}(\Omega)$ and $h(k+1)=h(1)$. Hence, by the definition of $h$ in \eqref{eq:hx}, the inequality in \eqref{eq:decreasingd} becomes
\begin{equation*}
0 \leq \min_{y \in \partial_\X \mathbb{C}(\Omega)} d(y, \Omega) - \min\{d(y,z) \colon y \in \partial_\X \Omega, z \in \partial_\X \mathbb{C}(\Omega)\},
\end{equation*}
which holds, because $\min\{d(y,z) \colon y \in \partial_\X \Omega, z \in \partial_\X \mathbb{C}(\Omega)\}=\min_{y \in \partial_\X \mathbb{C}(\Omega)} d(y, \Omega)$. Consider now a generic $k \geq 1$.
Since $\Omega$ and $\mathbb{C}^{k+1}(\Omega)$ are compact sets, there exist $\bar x\in\Omega$ and $\bar z\in\partial_\X\mathbb{C}^{k+1}(\Omega)$ such that $d(\bar x,\bar z) = \min_{y \in \partial_\X \mathbb{C}^{k+1}(\Omega)} d(y, \Omega)$. Since $\Omega\subset\inte_\X\left(\mathbb{C}^k(\Omega)\right)\subset \inte_\X\left(\mathbb{C}^{k+1}(\Omega)\right)$ and $\mathbb{C}^k(\Omega)$ is also a compact set, there exists $\bar y\in\partial_\X \mathbb{C}^k(\Omega)$ such that $d(\bar x,\bar z) = d(\bar x,\bar y) + d(\bar y,\bar z)$. Then,
    \begin{align*}
        \min_{y \in \partial_\X \mathbb{C}^{k}(\Omega)} d(y, \Omega) &\leq d(\bar x,\bar y)= \min_{y \in \partial_\X \mathbb{C}^{k+1}(\Omega)} d(y, \Omega) - d(\bar y,\bar z)\\
        &\leq \min_{y \in \partial_\X \mathbb{C}^{k+1}(\Omega)} d(y, \Omega) - h({k+1}),
    \end{align*}
for the $h$-function defined in~\eqref{eq:hx}.
\end{proof}

We define the following function on the set $\mathbb{C}_\infty(\Omega)$ in \eqref{eq:D_Omega}.

\begin{definition}\textbf{($\kappa-$function)}\label{def:kappa_function}
Given a RCIS $\Omega\subset\X$ for system~\eqref{eq:SistOrig}, the associated $\kappa$-function $\kappa \colon \mathbb{C}_\infty(\Omega)\to \N$ is
\[\kappa(x) = \min \{ k\in\N \colon x \in \mathbb{C}^k(\Omega) \}, \]
with $\kappa(x)=0$ for $x\in\Omega$.
\end{definition}

Next, we propose a constructive method to obtain an AG-function associated with a compact RCCS. 
Note that an AG-function is defined in the whole domain $\X$, but is guaranteed to be strictly decreasing along the system trajectories only in (a subset of) $\mathbb{C}_\infty(\Omega)\setminus\Omega$, associated with some $R>0$,
which is guaranteed to exist if $\Omega$ is a RCCS, since $\Omega\subset\inte_\X\big(\mathbb{C}(\Omega)\big)$,
and hence $\mathbb{C}_\infty(\Omega)\setminus\Omega\not=\emptyset$.

\begin{proposition}\label{prop:pseudoLyapunov}
Let $\Omega\subset \X$ be a RCCS for system~\eqref{eq:SistOrig}. The set-based map $L \colon \X\to\mathbb{R}_{\ge 0}$ defined for $x \in \mathbb{C}_\infty(\Omega)$ as
\begin{equation}\label{eq:L-Omega}
 L(x) = \min_{y \in \partial_\X \mathbb{C}^{\kappa(x)}(\Omega)} d(y, \Omega),
 \end{equation}
where $\kappa(x)$ is the $\kappa$-function associated with $\Omega$ as per Definition~\ref{def:kappa_function}, and for $x \in \X \setminus \mathbb{C}_\infty(\Omega)$ as $L(x)= \bar L$, for a positive $\bar L \ge \sup_{y \in \mathbb{C}_\infty(\Omega)} L(y)$, is an AG-function for system~\eqref{eq:SistOrig} and set $\Omega$ according to Definition~\ref{def:L-function}.
\end{proposition}
\begin{proof}
We prove that the function $L$ defined in \eqref{eq:L-Omega} satisfies both properties in Definition~\ref{def:L-function}.

Property \textbf{(i)}. For $x\in\Omega$, we have $\kappa(x)=0$.  Since $\Omega$ is closed, $\partial_\X\Omega\subseteq\Omega$, which implies that $d(z,\Omega)=0$ for all $z\in \partial_\X\mathbb{C}^0(\Omega) = \partial_\X\Omega$, and hence $L(x)=0$. 
Conversely, $\kappa(x)\ge 1$ for all $x\in\mathbb{C}_\infty(\Omega)\setminus\Omega$, and hence 
    \[\min_{y \in \partial_\X \mathbb{C}^{\kappa(x)}(\Omega)} d(y, \Omega)>0,\]
    since $\Omega\subset \inte_\X(\mathbb{C}^{\kappa(x)}(\Omega))$ for all $\kappa(x)\ge 1$ (see Proposition~\ref{propo:cont_nesetedC}); then, $L(x)>0$ when $x\in\mathbb{C}_\infty(\Omega)\setminus\Omega$. Finally, by definition, when $x \in \X \setminus \mathbb{C}_\infty(\Omega)$, $L(x)= \bar L > 0$.
    
Property \textbf{(ii)}. Consider $x \in \mathbb{C}_\infty(\Omega)$. Given the $\kappa$-function $\kappa(x)$ associated with $\Omega$ as per Definition~\ref{def:kappa_function}, we have that $x\in\mathbb{C}^{\kappa(x)}(\Omega)$. By the definition of robust controllable sets and the nested structure, there is $\sigma_x\in\N_q$ such that
    \begin{equation*}
    \boldsymbol{\Phi}(x, \sigma_x)\subset \mathbb{C}^{\kappa(x)-1}(\Omega).
    \end{equation*}
    Then, for all $z\in\boldsymbol{\Phi}(x, \sigma_x)$, we have that $\kappa(z) \leq \kappa(x)-1$, and hence
    \begin{equation*}
    \min_{y\in\partial_\X \mathbb{C}^{\kappa(z)}(\Omega)}d(y,\Omega) \le \min_{y\in\partial_\X \mathbb{C}^{\kappa(x)-1}(\Omega)}d(y,\Omega).
    \end{equation*}
    Since the last inequality holds pointwise in $z$, then
    \begin{equation}\label{eq:theorem1}
    \sup_{z\in\boldsymbol{\Phi}(x, \sigma_x)} \min_{y\in \partial_\X\mathbb{C}^{\kappa(z)}(\Omega)} d(y,\Omega) \le \min_{y\in \partial_\X \mathbb{C}^{\kappa(x)-1}(\Omega)}d(y,\Omega).
    \end{equation}
    Now consider that:
    \begin{align}
    \min_{\sigma^1\in\mathbb{N}_q} \sup_{z\in\boldsymbol{\Phi}(x, \sigma^1)} L(z)&\le \sup_{z\in\boldsymbol{\Phi}(x, \sigma_x)} L(z) \label{eq:theorem3}\\
    &= \sup_{z\in\boldsymbol{\Phi}(x, \sigma_x)} \left( \min_{y\in \partial_\X \mathbb{C}^{\kappa(z)} (\Omega)}d(y,\Omega)\right)\label{eq:theorem4} \\
    &\le \min_{y\in \partial_\X \mathbb{C}^{\kappa(x)-1}(\Omega)}d(y,\Omega) \label{eq:theorem5}\\
    &\le \min_{y\in \partial_\X \mathbb{C}^{\kappa(x)}(\Omega)}d(y,\Omega) - h(\kappa(x))\label{eq:theorem6}\\
    &= L(x) - \mathbf{h}(x),\label{eq:theorem7}
    \end{align}
    where $\mathbf{h}(x) := h(\kappa(x))$, and $h(\kappa(x))$ is given by \eqref{eq:hx}. Note that inequality \eqref{eq:theorem5}  follows from~\eqref{eq:theorem1}, while inequality \eqref{eq:theorem6} follows from~\eqref{eq:decreasingd} in Proposition~\ref{propo:decreasingd}.

    Since $\Omega$ is a RCCS, $\mathbb{C}_\infty(\Omega) \setminus\Omega\not=\emptyset$, and there exists $R>0$ such that $N_R(\Omega):=\{x \in \X \colon d(x,\Omega)\le R\} \subseteq \mathbb{C}_\infty(\Omega)$. Now, for any $r$ with $0\le r\le R$, we define the ``ring'' set $N_{r,R}(\Omega):=\{x \in \X \colon r\le d(x,\Omega)\le R\}$. The family $ \{\inte_\X\left(\mathbb{C}^i(\Omega)\right)\}_{i\ge 1}$ is an open cover of $N_{r,R}(\Omega)$. Since $N_{r,R}(\Omega)$ is compact, there is a finite subcover 
    \[ N_{r,R}(\Omega) \subset \bigcup_{j=1}^m \inte_\X\left(\mathbb{C}^{i_j}(\Omega)\right)\]
    and
    \[\min_{x\in N_{r,R}(\Omega)} \mathbf{h}(x)=\min_{i\in S_m} h(i)>0,\] 
    for the index set $S_m = \{i_1,\dots,i_m\}$.

    Define $\eta_R(0)=0$ and, for each $r\in(0,R]$,
    \[
    \eta_R(r)\;:=\;\min_{x\in N_{r,R}(\Omega)} \mathbf{h}(x).
    \]
    Since $N_{r_2,R}(\Omega)\subseteq N_{r_1,R}(\Omega)$ whenever $0\le r_1\le r_2\le R$, function $\eta_R(r)$ is non-decreasing. Moreover, any $x\in N_{R}(\Omega)$ also satisfies $x\in N_{d(x,\Omega),R}(\Omega)$, and hence
    \[
    \mathbf{h}(x)\ \ge\ \eta_R\big(d(x,\Omega)\big).
    \]
    Combining the above inequality with \eqref{eq:theorem3}–\eqref{eq:theorem7} yields
    \[
    \min_{\sigma^1\in\mathbb{N}_q}\ \sup_{z\in \boldsymbol{\Phi}(x,\sigma^1)}L(z)-L(x)
    \le -\eta_R\!\big(d(x,\Omega)\big).
    \]
    The inequality also holds for all $x \in \mathcal{L}_{\hat R}$, where $\mathcal{L}_{\hat R}$ is the largest level set of $L$ such that $\mathcal{L}_{\hat R} \subseteq N_R(\Omega)$, which concludes the proof.
\end{proof}

The AG-function proposed in Proposition~\ref{prop:pseudoLyapunov} is piecewise constant, as it is constant over each set $\mathbb{C}^k(\Omega)\setminus\mathbb{C}^{k-1}(\Omega)$.

\begin{theorem}\label{theo:RCCS_AG}
If $\Omega$ is a RCCS for system~\eqref{eq:SistOrig}, then $\Omega$ is robustly locally attractive under the system dynamics.
Moreover, an explicit stabilizing switching law is 
\begin{equation}\label{eq:switchinglaw}
\hspace{-1mm} \sigma(x)\ = \begin{cases}
\min\big\{\sigma^1 \in\mathbb{N}_q \colon x\in \mathbb{C}\big(\sigma^1,\mathbb{C}^{\kappa(x)-1}(\Omega)\big)\big\}\\ \qquad \mbox{ for } x \in \mathbb{C}_\infty(\Omega)\setminus \Omega,\\[2mm]
\min\big\{\sigma^1 \in\mathbb{N}_q \colon x\in \mathbb{C}\big(\sigma^1,\Omega\big)\big\}\\ \qquad  \mbox{ for } x \in \Omega,
\end{cases}
\end{equation}
where $\mathbb{C}\big(\sigma^1,\mathbb{C}^{\kappa(x)-1}(\Omega)\big)$ is the robust controllable set to set $\mathbb{C}^{\kappa(x)-1}(\Omega)$ for mode $\sigma^1$, as per Definition~\ref{def:1-scssigma}.
\end{theorem}
\begin{proof}
Robust local attractivity of $\Omega$ follows from combining Proposition~\ref{prop:pseudoLyapunov}, by which the existence of a RCCS guarantees that a suitable AG-function exists, and Proposition~\ref{prop:L-function-implies-attractivity}, by which the existence of a suitable AG-function guarantees robust local attractivity.

We now show that $\sigma(x)$ in \eqref{eq:switchinglaw} is a stabilizing switching law. By Definition~\ref{def:kappa_function} of function $\kappa$, we have
\[x\in\mathbb{C}^{\kappa(x)}(\Omega)=\mathbb{C}\big(\mathbb{C}^{\kappa(x)-1}(\Omega)\big) = \bigcup_{\sigma^1\in\N_q}\mathbb{C}\big(\sigma^1,\mathbb{C}^{\kappa(x)-1}(\Omega)\big),\]
hence the set in \eqref{eq:switchinglaw} is non-empty; taking the minimum makes $\sigma(x)$ single-valued. 
By Definition~\ref{def:1-scssigma}, it is $\boldsymbol{\Phi}(x,\sigma(x))\subseteq\mathbb{C}^{\kappa(x)-1}(\Omega)$; thus, for every $z\in\boldsymbol{\Phi}(x,\sigma(x))$, $\kappa(z)\le \kappa(x)-1$. Consider $\bar \kappa:=\kappa(x)- \kappa(z)\ge 1$. Then, considering the AG-function defined in \eqref{eq:L-Omega} and relying on the inequality \eqref{eq:decreasingd} in Proposition~\ref{propo:decreasingd}, we have that 
\begin{align*}
L(z) &= \min_{y\in \partial_\X \mathbb{C}^{\kappa(z)}(\Omega)}d(y,\Omega) \\
&\le \min_{y\in \partial_\X \mathbb{C}^{\kappa(x)}(\Omega)}d(y,\Omega) - \sum_{i=1}^{\bar \kappa} h(\kappa(z)+i) \\
& \le L(x) - h(\kappa(x)),
\end{align*}
for all $z\in\boldsymbol{\Phi}(x,\sigma(x))$, which is the decrease condition in \eqref{eq:theorem3}–\eqref{eq:theorem7}, and thus, as shown before, guarantees robust local attractivity of $\Omega$ under the switching law $\sigma(x)$.
\end{proof}
\begin{remark}
Our result holds for generic nonlinear switched systems with bounded uncertainties. However, a limitation in the applicability of the switching law $\sigma(x)$ in \eqref{eq:switchinglaw} is that its implementation requires the construction of the controllable sets $\mathbb{C}^k(\Omega)$, $k=1,\dots,N$, with $N$ large enough to ensure that $x \in \mathbb{C}^N(\Omega)$. In fact, to compute $\sigma(x)$, one must check whether $x \in \mathbb{C}(\sigma^1,\mathbb{C}^{\kappa(x)-1}(\Omega))$ for some $\sigma^1$; this in turn requires determining $\kappa(x)$, which is obtained by finding the minimum index $k=0,1,\dots,N$ such that $x \in \mathbb{C}^k(\Omega)$.
For general nonlinear systems, given a RCCS $\Omega$, computing the set $\mathbb{C}^N(\Omega)$ is challenging, but then the switching law can be determined offline and applied without additional computational effort for all $x \in \mathbb{C}^N(\Omega)$.
In Section~\ref{sec:linear} we present an explicit method to construct the state-dependent switching law $\sigma(x)$ in~\eqref{eq:switchinglaw} for \textit{linear} switched systems with additive uncertainties.
\end{remark}

\begin{remark}
The construction of the set-based AG-function in Proposition~\ref{prop:pseudoLyapunov} is inherently conservative, since it relies on identifying a RCCS, which, as mentioned, is challenging for a generic nonlinear system with a generic uncertainty structure. Exploring alternative approaches to establish the existence of an AG-function (or construct one), beyond the contractive-set-based construction in Proposition~\ref{prop:pseudoLyapunov}, is therefore a direction for future work.
\end{remark}

\begin{figure}[ht]
	\centering	\includegraphics[width=0.85\linewidth]{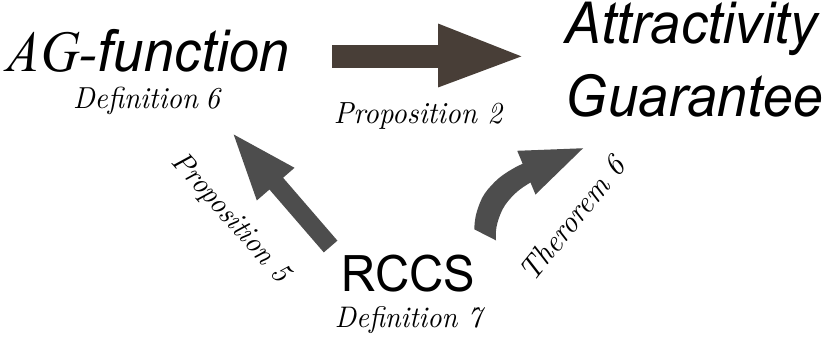}
	\caption{\small Scheme that illustrates the structure of our main results. The existence of an AG-function guarantees local attractivity of a set for system~\eqref{eq:SistOrig} (Proposition~\ref{prop:L-function-implies-attractivity}). The construction of a RCCS ensures the existence of a particular AG-function (Proposition~\ref{prop:pseudoLyapunov}). Then, the existence of a RCCS is sufficient to conclude local attractivity of the set (Theorem~\ref{theo:RCCS_AG}).}
	\label{fig:diagram}
\end{figure}

The proposed framework establishes local set-attractivity, but does not provide global attractivity or stability properties of the system.
Conversely, for \textit{uncertainty-free} \textit{linear} switched systems, the existence of a RCCS may imply global attractivity under suitable homogeneity assumptions; see~\cite{fiacchini2014necessary}.

Next, we briefly examine the specificities of the linear case in Section~\ref{sec:linear} and then, in Section~\ref{sec:AMR}, we consider the case study of control of nonlinear switching dynamics in the context of antimicrobial resistance, illustrating that the region of local attractivity, $\mathcal{L}_R(\Omega)$, can be considerably larger than the obtained RCCS $\Omega$. 

%%%%%%%%%%%%%%%%%%%%%%%%%%%%%%%%%%%%%%%%%%%%%%%%%%%%%%%%%%%%%%%%%%%%%%%%%%%%%%%%%%%%%%%%%%%%%%%%%%%%%%%%%
\subsection{Linear case}\label{sec:linear}
%%%%%%%%%%%%%%%%%%%%%%%%%%%%%%%%%%%%%%%%%%%%%%%%%%%%%%%%%%%%%%%%%%%%%%%%%%%%%%%%%%%%%%%%%%%%%%%%%%%%%%%%%

For a linear switched system with additive uncertainty
\begin{equation}\label{eq:linearss}
x(k+1) = A_{\sigma(k)} x(k) + w(k),
\end{equation}
where $x(k) \in \X$, $w(k) \in \W$, sets $\mathbb{X}$ and $\mathbb{W}$ are compact and contain the origin in their interior, $\sigma(k) \in \mathbb{N}_q$ is the active mode at time $k$, selected according to the switching law $\sigma \colon \N\rightarrow \mathbb{N}_q$, and all matrices $A_{\sigma(k)}$, $\sigma(k) \in \mathbb{N}_q$, are nonsingular, we can rely on a specific characterisation of robust controllable sets. 

Consider the Minkowski (Pontryagin) difference of two sets $\mathcal{A}\subset\X$ and $\mathcal{B}\subset\X$,
\begin{equation*}
\mathcal{A} \ominus \mathcal{B} = \{a \colon a + b \in \mathcal{A}, \text{ for all } b \in \mathcal{B}\}.
\end{equation*} 

Given a set $\Omega \subseteq \X$, with $\Omega \ominus \W \neq \emptyset$, the robust $1$-step controllable set to $\Omega$ (as per Definition~\ref{def:1-scs}) can be expressed as
\begin{equation}\label{eq:C_linear_case}
  \mathbb{C}(\Omega)= \bigcup\limits_{\sigma^1 \in \mathbb{N}_q} A^{-1}_{\sigma^1} \left(\Omega \ominus \W\right) \cap \X,
\end{equation}
since, if $x\in \bigcup_{\sigma^1} A^{-1}_{\sigma^1} \left(\Omega \ominus \W\right)$, there exists $\sigma^1\in\mathbb{N}_q$ such that $x\in A^{-1}_{\sigma^1} \left(\Omega \ominus \W\right)$, which implies $A_{\sigma^1} x \in \Omega \ominus \W$. Hence, $A_{\sigma^1}x+w\in\Omega$ for all $w\in\W$, meaning that $\boldsymbol{\Phi}(x,\sigma^1)\subseteq \Omega$. Hence, $x\in\mathbb{C}(\Omega)$. Conversely, if $x\in\mathbb{C}(\Omega)$, there must exist $\sigma^1\in\N_q$ such that $A_{\sigma^1} x + w \in \Omega$ for all $w\in\W$, or equivalently $A_{\sigma^1} x \in \Omega\ominus\W$, and thus
\[ x\in A^{-1}_{\sigma^1}(\Omega\ominus\W) \cap \X \subseteq \bigcup\limits_{\sigma^1 \in \mathbb{N}_q} A^{-1}_{\sigma^1} \left(\Omega \ominus \W\right) \cap \X. \]

\begin{remark}
For compact and convex sets $\Omega$ and $\mathbb{W}$, both containing the origin in their interiors,
$\mathbb{C}(\Omega)$ is a $C^*$-set (a compact, convex set that contains the origin in its interior and is star-convex with respect to the origin). Such sets can be treated using standard geometric tools from linear and convex analysis \citep{herceg2013multi}.
\end{remark}

Building upon \eqref{eq:C_linear_case}, we can extend the algorithm proposed by \cite{fiacchini2014necessary} for uncertainty-free linear switched systems (with $w = 0$), in a way that enables us to also account for bounded uncertainty $w \in \mathbb{W}$.
\begin{myalgorithm}\label{alg:Procedure_1}
    \item[i.] Set $k = 1$. Select $k_{stop}\in\N$ and a RCIS $\Omega_0 \subset \X$.
    \item[ii.] Compute $\Omega_{k} = \bigcup_{\sigma^1 \in \mathbb{N}_q} A^{-1}_{\sigma^1} \left(\Omega_{k-1} \ominus \W\right) \cap \mathbb{X}$
    \item[iii.] If $\Omega_0 \subseteq \inte_\X\Big(\bigcup_{j=1}^{k} \Omega_j\Big)$, stop successfully.
    \item[iv.] Else, if $k=k_{stop}$, stop unsuccessfully. 
    \item[v.] Otherwise, set $k \leftarrow k+1$ and go to step ii. 
\end{myalgorithm}

Algorithm~\ref{alg:Procedure_1} recursively computes the robust $k$-step controllable sets $\Omega_k=\mathbb{C}^k(\Omega_0)$ to the RCIS $\Omega_0$.  The algorithm ensures that $\Omega_1=\mathbb{C}(\Omega_0)$ is a RCCS if it stops successfully when $k=1$. Conversely, if the algorithm stops successfully when $k = N > 1$, then it guarantees that the RCIS $\Omega_0$ is robustly locally attractive and provides an inner approximation $\bigcup_{j=1}^{N} \Omega_j$ of $\mathbb{C}_\infty(\Omega_0)$ that allows us to compute offline the switching law \eqref{eq:switchinglaw}, which is well defined for all $x \in \bigcup_{j=1}^{N} \Omega_j$.

\begin{remark}\label{rem:Initial_Set_Omega_0}
The outcome of the algorithm critically depends on the choice of an initial RCIS $\Omega_0$. Recent approaches to characterise and construct invariant sets in this setting have been provided by \cite{danielson2019necessary,perez2025characterization}. 
As long as $\Omega_0$ is a RCIS, the condition at step ii. of Algorithm~\ref{alg:Procedure_1} could be equivalently replaced by $\Omega_0 \subseteq \inte_\X \Omega_k$, since $\bigcup_{j=1}^{k} \Omega_j = \Omega_k$ in view of Proposition~\ref{proper:nestedC}. However, considering the condition $\Omega_0 \subseteq \inte_\X\Big(\bigcup_{j=1}^{k} \Omega_j\Big)$ allows us to use the very same algorithm structure to find a RCIS: if an arbitrary convex set $\tilde\Omega_0$ is selected at step i. of Algorithm~\ref{alg:Procedure_1} and the algorithm stops successfully at step $k = \hat N$, then $\hat \Omega_{\hat N} := \bigcup_{j=1}^{\hat N} \tilde \Omega_j$, where $\tilde \Omega_{j} = \bigcup_{\sigma^1 \in \mathbb{N}_q} A^{-1}_{\sigma^1} \left(\tilde \Omega_{j-1} \ominus \W\right) \cap \mathbb{X}$, is a RCIS \citep{anderson2024stabilizability},
and we can subsequently apply Algorithm~\ref{alg:Procedure_1} to the RCIS $\Omega_0:=\hat \Omega_{\hat N}$ to assess its contractivity.
\end{remark}

\begin{remark}\label{rem:switchlaw}
\textbf{Computing the switching law $\sigma(x)$.} Given a RCIS $\Omega_0$, Algorithm~\ref{alg:Procedure_1} computes recursively the controllable sets 
\begin{equation*}
\mathbb{C}^k(\Omega_0)=\Omega_k =
\bigcup_{\sigma^1 \in \mathbb{N}_q}
\underbrace{
A^{-1}_{\sigma^1}\big(\Omega_{k-1}\ominus \W\big)
\cap \X
}_{\text{finite union of convex polytopes}} 
\end{equation*}
for $k=1,\dots,N$, needed to implement $\sigma(x)$ in~\eqref{eq:switchinglaw} for $x \in \Omega_N \setminus \Omega_0$.
Assuming that the initial set $\Omega_0$ and the sets $\X$ and $\W$ are convex polytopes, each set $\Omega_k$ can be represented as a finite union of convex polytopes, each admitting an $H$-representation of the form $\{x : Mx \le b\}$.
Once these sets are available, the computation of $\sigma(x)$ for $x \in \Omega_N \setminus \Omega_0$ requires two steps.
\begin{enumerate}
\item[I.] Determine the value of $\kappa(x)$, defined as the minimum index $k=1,\dots,N$ such that $x \in \Omega_k$. Since $\Omega_k$ is a finite union of polytopes, computing $\kappa(x)$ reduces to checking membership in at least one of the polytopes, each associated with a system of linear inequalities. \item[II.] Select the smallest $\sigma^1 \in \mathbb{N}_q$ such that
\[
x \in \mathbb{C}(\sigma^1, \mathbb{C}^{\kappa(x)-1}(\Omega_0))
=
A^{-1}_{\sigma^1}(\Omega_{\kappa(x)-1}\ominus\W)\cap\X,
\]
which again reduces to checking polytope membership.
\end{enumerate}

To compute $\sigma(x)$ for $x \in \Omega_0$, just select the smallest $\sigma^1 \in \mathbb{N}_q$ such that $x \in \mathbb{C}(\sigma^1, \Omega_0)
=
A^{-1}_{\sigma^1}(\Omega_0\ominus\W)\cap\X$.
Therefore, once the sets $\Omega_k$ have been computed offline using Algorithm~\ref{alg:Procedure_1}, the online implementation of $\sigma(x)$ consists solely of polytope membership tests.
\end{remark}

The next two examples illustrate how our approach particularises for linear switched systems \eqref{eq:linearss}.

\begin{example}\label{example1}
Consider a system of the form~\eqref{eq:linearss} with two modes associated with non-Hurwitz matrices
\begin{equation*}
A_1 = \begin{bsmallmatrix} 0.3 & -1.01 \\ -0.5 & -0.8 \end{bsmallmatrix}, \quad
A_2 = \begin{bsmallmatrix} -0.4 & 1.2 \\ 0.9 & -0.5 \end{bsmallmatrix},
\end{equation*}
set $\X = \{ (x_1, x_2) \in \mathbb{R}^2 \colon  | x_i | \leq 6, \,\, i=1,2 \}$, and set $\W = \{ (w_1, w_2) \in \mathbb{R}^2 \colon  | w_i | \leq 0.1, \,\, i=1,2 \}$.

To identify a RCIS, as discussed in Remark~\ref{rem:Initial_Set_Omega_0}, we first run Algorithm~\ref{alg:Procedure_1} initialised with the convex set $\tilde\Omega_0$ chosen as a polytope with $100$ faces that approximates the set $\{ x \in \X \colon  \| x \| \leq 1 \}$.
Since $\tilde\Omega_0\not \subseteq \tilde\Omega_1$, the set $\tilde\Omega_0$ is not a RCIS. However, after $3$ iterations, $\tilde\Omega_0 \subset \bigcup_{j=1}^{3}\tilde\Omega_j$, and hence the set $\hat\Omega_3:=\bigcup_{j=1}^{3}\tilde\Omega_j$ is a RCIS, formed by the union of $N_{sets} = 14$ polytopes (at each step $k$, we compute $N_{sets} = 2 (2^k - 1)$ polytopes).

To check contractivity of the RCIS $\Omega_0:=\hat\Omega_3$, we then run Algorithm~\ref{alg:Procedure_1} again with initial set $\Omega_0:=\hat\Omega_3$ and we reach $k_{stop}=3$, computing a total number of $112$ polytopes (since $\Omega_0$ is the union of $14$ polytopes).
Since the condition $\Omega_0\subseteq\inte_\X\big(\bigcup_{j=1}^3\Omega_j\big)$ is not fulfilled, the algorithm stops unsuccessfully.
The resulting inner approximation of $\mathbb{C}_\infty(\Omega_0)$ is given by $\tilde{\mathbb{D}}(\Omega_0) :=\bigcup_{j=1}^3\Omega_j$, which by construction is also an approximation of $\mathbb{C}_\infty(\tilde\Omega_0)$, because all trajectories originating from some $x \in \Omega_0$ can be driven to $\tilde\Omega_0$ in three steps. Once a trajectory enters $\tilde{\Omega}_0$, it can never escape the RCIS $\Omega_0 \supseteq \tilde{\Omega}_0$.
Since the approach does not identify a RCCS, our framework cannot be applied to guarantee the existence of an AG-function; however, the approach still allows us to compute offline, as discussed in Remark~\ref{rem:switchlaw}, the switching law $\sigma(x)$ in \eqref{eq:switchinglaw} that drives to the set $\tilde{\Omega}_0$ all trajectories emanating from $\tilde{\mathbb{D}}(\Omega_0)$, and keeps them inside the RCIS $\Omega_0$.

Fig.~\ref{fig:LinearExample1}-A shows the sets $\tilde\Omega_0$ (black contour); the disturbance set $\W$ (gray); the set $\tilde\Omega_0\ominus\W$ (orange), and the Pontryagin difference in the inset; the set $\Omega_0$ (teal);  and the set $\tilde{\mathbb{D}}(\Omega_0)$ (light blue). As expected, we have $\tilde \Omega_0 \subset \Omega_0 \subset \tilde{\mathbb{D}}(\Omega_0)$.
Fig.~\ref{fig:LinearExample1}-B shows the closed-loop trajectories of the system $x(k+1) = A_{\sigma(x(k))}x(k) + w(k)$ emanating from the initial condition $x_0 = (-3, 1.2) \in \tilde{\mathbb{D}}(\Omega_0)$. The switching law $\sigma(x(k))$, which selects the active subsystem $A_{\sigma(x(k))}$ according to \eqref{eq:switchinglaw}, is shown in Fig.~\ref{fig:LinearExample1}-C. A total of $100$ controlled trajectories are simulated to evaluate the robustness of the switching law $\sigma(x)$ under additive uncertainty $w \in \mathbb{W}$. At each time step, the uncertainty $w$ is randomly drawn from the compact set $\mathbb{W}$ with a uniform probability distribution.

Considering a larger $k_{stop}=6$ leads to the computation of a total number of $1008$ polytopes and the inner approximation $\bigcup_{j=1}^6\Omega_j$ for $\mathbb{C}_\infty(\Omega_0)$; however, also in this case
$\Omega_0\subseteq\inte_\X\big(\bigcup_{j=1}^6\Omega_j\big)$ is not fulfilled, and the algorithm stops unsuccessfully, although a suitable switching law $\sigma(x)$ can be constructed for all $x \in \bigcup_{j=1}^6\Omega_j$.
\end{example}
\begin{figure}[ht]
	\centering
	\includegraphics[width=0.85\linewidth]{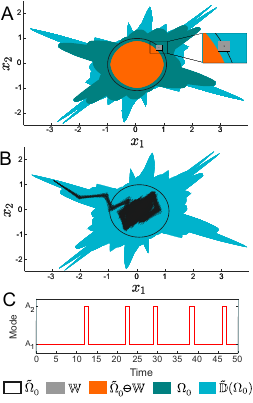}
	\caption{\small Example~\ref{example1}. \textbf{A.} The boundary of $\tilde\Omega_0$ is denoted by the black contour. Set $\W$ is in gray, and the Pontryagin difference $\tilde\Omega_0 \ominus \mathbb{W}$ in orange (highlighted in the inset). The set $\Omega_0$ is in teal and the set $\tilde{\mathbb{D}}(\Omega_0):=\bigcup_{j=1}^{3} \Omega_j$ in light blue. \textbf{B.} Controlled trajectories under the switching law $\sigma(x(k))$, for $100$ different random uncertainty sequences $w^{50} \in \mathbb{W}^{50}$. \textbf{C.} The switching law $\sigma(x(k))$ from \eqref{eq:switchinglaw}.}
	\label{fig:LinearExample1}
\end{figure}

Algorithm~1 by \cite{fiacchini2014necessary}, designed to establish stability of a set in the \textit{uncertainty-free} setting, can be applied to the same system \eqref{eq:linearss} with matrices $A_1$ and $A_2$, without uncertainties ($w=0$), and with the same initial set $\tilde{\Omega}_0$. Their algorithm converges in two iterations, thereby establishing the stability of $\tilde{\Omega}_0$ for the uncertainty-free linear switched system. This contrasts with our results, where robust attractivity could not be established due to the presence of the uncertainties.

\begin{example}\label{example2}
Now consider a system of the form \eqref{eq:linearss} with $\X = \{ (x_1, x_2) \in \mathbb{R}^2 \colon  -3\le x_1 \leq 4,~ -10\le x_2 \leq 10 \}$, $\W = \{ w \in \mathbb{R}^2 \colon  \| w \| \leq 0.05 \}$, and four modes associated with matrices
\begin{align*}
A_1 &= \begin{bsmallmatrix} -0.3912 & 0.9743 \\ -1.0409 & 0.1366 \end{bsmallmatrix}, \quad
A_2 = \begin{bsmallmatrix} 0.0609 & 1.0481 \\ -0.8837 & 0.5669 \end{bsmallmatrix}, \quad \\
A_3 &= \begin{bsmallmatrix} 0.9743 & 0.3912 \\ 0.1366 & 1.0409 \end{bsmallmatrix}, \quad
A_4 = \begin{bsmallmatrix} -1.0481 & 0.0609 \\ -0.5668 & -0.8837 \end{bsmallmatrix}.
\end{align*}

We run Algorithm~\ref{alg:Procedure_1} initialised with a convex set $\Omega_0$ chosen as a polytope with $100$ faces that approximates the set $\{ x \in \X \colon  \| x \| \leq 1 \}$.  The algorithm stops successfully after one iteration, as shown in Fig.~\ref{fig:LinearExample2}.
In fact, $\Omega_0\subset \inte_\X\big(\Omega_1)$, where $\Omega_1=\mathbb{C}(\Omega_0)$, and hence $\Omega_0$ is a RCIS and $\Omega_1$ is a RCCS. Also, for $\Omega_2=\mathbb{C}^2(\Omega_0)$, we have $\Omega_1\subset \inte_\X(\Omega_2)$, consistently with Proposition~\ref{propo:cont_nesetedC}. 
Since $\Omega_1$ is a RCCS, an AG-function exists with respect to $\Omega_1$ (Proposition~\ref{prop:pseudoLyapunov}), which is therefore \textit{robustly locally attractive} (Theorem~\ref{theo:RCCS_AG}). The set $\tilde{\mathbb{D}}(\Omega_0) := \bigcup_{j=1}^6 \Omega_j$, given by the union of $5460$ polytopes, is an inner approximation of $\mathbb{C}_\infty(\Omega_0)$; by construction, it is also an approximation of $\mathbb{C}_\infty(\Omega_1)$. The switching law~\eqref{eq:switchinglaw}, starting from any point in $\tilde{\mathbb{D}}(\Omega_0)$, drives the system to the RCCS $\Omega_1$ by steering it to the RCIS $\Omega_0$, which is contained in $\Omega_1$.
\end{example}
\begin{figure}[ht]
	\centering
	\includegraphics[width=0.85\linewidth]{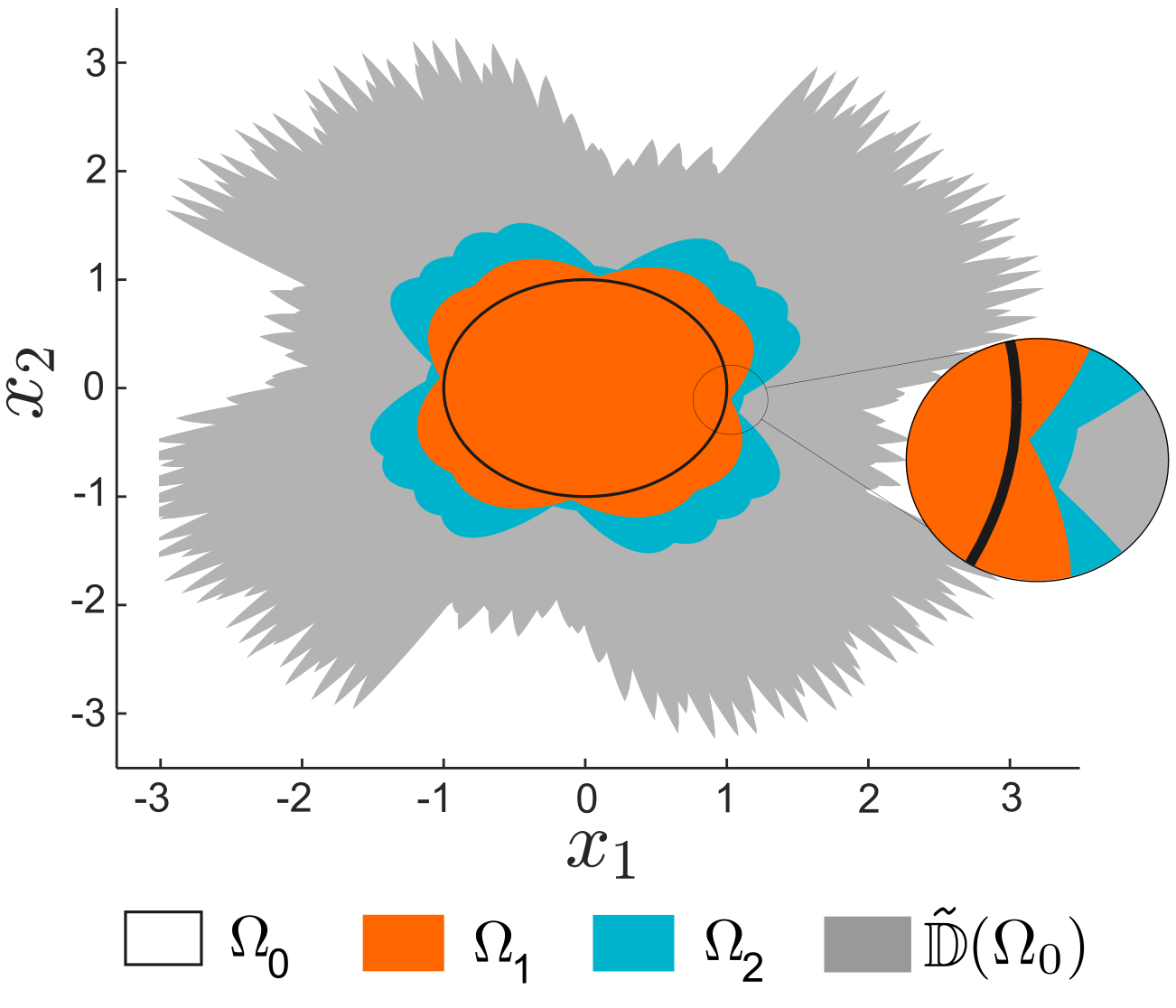}
	\caption{\small  Example~\ref{example2}. Given $\Omega_0$, whose boundary is denoted by the black contour, the set $\Omega_1=\mathbb{C}(\Omega_0)$ in orange is a RCCS. For the blue set $\Omega_2=\mathbb{C}^2(\Omega_0)$, it holds $\Omega_1\subset \inte_\X\big(\Omega_2\big)$ as expected. Set $\mathbb{C}_\infty(\Omega_0)$ is approximated by $\tilde{\mathbb{D}}(\Omega_0) := \bigcup_{j=1}^6 \Omega_j$ in grey.}
	\label{fig:LinearExample2}
\end{figure}

%%%%%%%%%%%%%%%%%%%%%%%%%%%%%%%%%%%%%%%%%%%%%%%%%%%%%%%%%%%%%%%%%%%%%%%%%%%%%%%%%%%%%%%%%%%%%%%%%%%%%%%%%
\section{Nonlinear Case Study: AMR Dynamics}\label{sec:AMR}
%%%%%%%%%%%%%%%%%%%%%%%%%%%%%%%%%%%%%%%%%%%%%%%%%%%%%%%%%%%%%%%%%%%%%%%%%%%%%%%%%%%%%%%%%%%%%%%%%%%%%%%%%

Consider the nonlinear continuous-time dynamic model proposed by \cite{anderson2025failure} to capture the time evolution of antibiotic-susceptible and antibiotic-resistant bacterial populations in the host, in the presence of immune-system response and antibiotics:
\begin{equation}\label{eq:b-s-model}
\begin{cases}
\dot b = \alpha\, b\!\left(1-\tfrac{b}{N}\right) - I(b)\,b - D(u)\,s \\
\dot s = \alpha\, s\!\left(1-\tfrac{b}{N}\right) - I(b)\,s - D(u)\,s - M(u)\,s,
\end{cases}
\end{equation}
where $s$ is the drug–susceptible subpopulation and $b$ denotes the total bacterial load given by susceptible and resistant subpopulations. The states are constrained to the feasible domain $\X := \{(b,s)\in\mathbb{R}^2:\ 0 \le s \le b \le N\}$. 
All subpopulations have net growth rate $\alpha>0$, affected by a  logistic term with carrying capacity $N>0$ (i.e., we are assuming there is no \emph{cost of resistance}, see \citealt{gandra2014economic}). The killing rate due to the host immune response is $I(b)=\beta\,\tfrac{K}{K+b}$, where $\beta>\alpha$ is the maximum killing rate and $K<N$ is the total bacterial load yielding half-maximal rate. The antibiotic killing rate is the Hill-type function $D(u)$, where $u$ is the antibiotic concentration, and the antibiotic-induced resistance rate is the threshold function $M(u)$; see \cite{anderson2025failure} for details. 

We consider here two modes with different constant antibiotic levels: immune response alone, with $u\equiv 0$ (mode~1); immune response plus antibiotic dosing, with $u\equiv u_M$ (mode~2), where $u_M$ is the Minimum Inhibitory Concentration (MIC) of the antibiotic \citep{kowalska2021minimum}.  
Denoting the state by $z=(b,s)\in \X$, we consider the continuous-time nonlinear switched system
\[
\dot z = s_\sigma(z), \quad \sigma\in\{1,2\},
\]
with 
\begin{align*}
s_1(z)&=
\begin{bsmallmatrix}
\alpha b(1-b/N) - I(b)\,b\\
\alpha s(1-b/N) - I(b)\,s
\end{bsmallmatrix},\\
s_2(z)&=
\begin{bsmallmatrix}
\alpha b(1-b/N) - I(b)\,b - D_M\,s\\
\alpha s(1-b/N) - I(b)\,s - D_M\,s - \mu\,s
\end{bsmallmatrix},
\end{align*}
where $D_M:=D(u_M)>\alpha$ is the constant death rate of susceptible bacteria exposed to antibiotic concentration $u_M$, while $\mu:=M(u_M)>0$ is the corresponding antibiotic-induced resistance rate.

We consider the discretisation of the continuous-time switched system obtained by applying the forward Euler method with sampling period $\Delta$:
\begin{align}\label{eq:Ex_system}
    x(k+1) &=\; x(k) + \Delta\, s_{\sigma(k)}(x(k)) + w(k) \\ \notag
    &:=\; f_{\sigma(k)}(x(k),w(k)),
\end{align}
where $s_{\sigma(k)}(x(k))$ is the active subsystem vector field at time $t_k=k\Delta$, and $w(k)=(w_b(k),w_s(k))$ represents the discretisation error, which lies in a bounded set $\W$ provided that the sampling step $\Delta$ is sufficiently small \citep{hairer1993solving}.
To the discrete-time nonlinear switched system in \eqref{eq:Ex_system}, we can then apply the  proposed methodology.

Following the criterion in Remark~\ref{rem:check_RCCS}, we design a numerical procedure to identify a RCCS for system~\eqref{eq:Ex_system}.
Generate a dense grid of initial conditions in the set
\[
D_0 = \{(b,s)\in\X : b \le b_{\max}\}
\]
with $b_{max}\le N$ and, for $b_0 < b_{max}$, define a candidate set 
\[
\Omega_0 = \{(b,s)\in\X : b \le b_0\}.
\]
For each grid point $x_i=(b_i,s_i)\in D_0$, we compute the one-step evolution of the system under both modes and three uncertainty realisations:
$w_1=(\bar w_b,0)$, $w_2=(0,\bar w_s)$, $w_3=(\bar w_b,\bar w_s)$,
where $\bar w_b$ and $\bar w_s$ are the maximum values of $w_b$ and $w_s$, respectively.

For each mode $\sigma^1 \in \N_q$, we define the set of next states
\begin{align*}
X^+_{i,\sigma^1} &:= 
\{\, f_{\sigma^1}(x_i,w_1),\; f_{\sigma^1}(x_i,w_2),\; f_{\sigma^1}(x_i,w_3)\,\}\\ & \subseteq \boldsymbol{\Phi}(x_i,\sigma^1).
\end{align*}
An initial condition $x_i$ is labelled as \textit{inside} the robust $1$-step controllable set $\mathbb{C}(\Omega_0)$ if
\begin{equation}\label{eq:Algorithm_Nonlinear}
X^+_{i,\sigma^1} \subseteq \Omega_0 \text{ for some } \sigma^1\in\{1,2\}.
\end{equation}

We construct a boundary of the set of points $x_i$ fulfilling condition~\eqref{eq:Algorithm_Nonlinear} by connecting the outermost admissible grid points along each horizontal slice, suitably ordered, with a closed polygonal curve. The region enclosed by this curve provides an approximation of $\mathbb{C}(\Omega_0)$. The criterion in Remark~\ref{rem:check_RCCS} can be used to check if the set $\mathbb{C}(\Omega_0)$ is a RCCS.

\begin{remark}
Considering only $w_1$, $w_2$, and $w_3$ is sufficient since they represent the worst-case uncertainties at the positive upper corners of $\W$, and the uncertainty appears linearly (additively) in the dynamics. 
Negative values of $w_b$ or $w_s$ yield smaller state increments, which in this case favour contractivity.
\end{remark}

For all our numerical experiments, we fix step size $\Delta = 0.1h$ and uncertainty bounds $\bar w_b = \bar w_s = 5$, while the model parameters are the same as in \cite{anderson2025failure}.

Fig.~\ref{fig:RCCS} illustrates the proposed methodology applied to three candidate sets $\Omega_0$, using a grid of $11,325$ points in $D_0$ ($b_{\max}=150,000$).
Fig.~\ref{fig:RCCS}-A shows the case $b_0 = 1,000$; the inclusion $\Omega_0 \subseteq \inte_\X\big(\mathbb{C}(\Omega_0)\big)$ holds, confirming the existence of a RCCS for system~\eqref{eq:Ex_system}.
If we increase to $b_0 = 10,000$, we also obtain a RCCS (see Fig.~\ref{fig:RCCS}-B). 
Conversely, for $b_0=120,000$, the inclusion 
$\Omega_0 \subseteq \inte_\X\big(\mathbb{C}(\Omega_0)\big)$ does not hold for this candidate $\Omega_0$ and the selected grid
(see Fig.~\ref{fig:RCCS}-C).

\begin{figure*}[ht]
	\centering
    \includegraphics[width=0.95\linewidth]{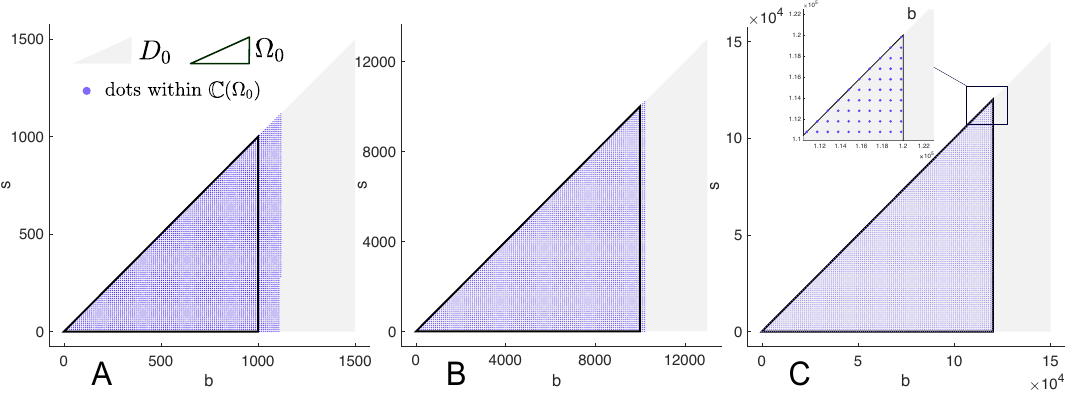}
	\caption{\small The grid in the set $D_0$ (with $b_{\max}\le 1.6\times10^6$) contains $11,325$ initial points uniformly distributed in the three cases. Blue points denote initial conditions $x_i\in D_0$ for which $X^+_{i,\sigma^1} \subseteq \Omega_0$ after one step for some $\sigma^1\in\{1,2\}$, to assess whether 
	$\Omega_0 \subseteq \inte_\X\big(\mathbb{C}(\Omega_0)\big)$ holds for the set $\Omega_0$ associated with the considered value $b_0$. For $b_0=1,000$ (\textbf{A.}) and $b_0=10,000$ (\textbf{B.}), the inclusion holds, and hence the considered $\Omega_0$ is a RCCS. For $b_0=120,000$ (\textbf{C.}), the inclusion does not hold for the adopted grid of points.}
	\label{fig:RCCS}
\end{figure*}

\begin{figure*}[ht]
	\centering
	\includegraphics[width=0.95\linewidth]{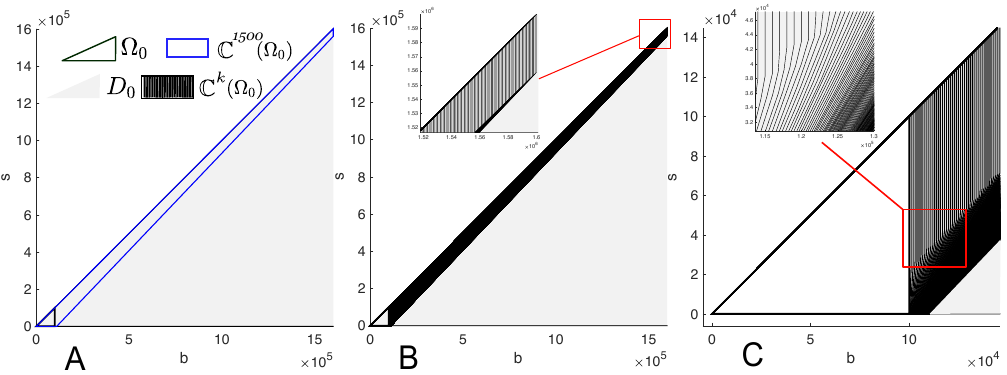}
	\caption{\small Approximation of $\mathbb{C}_\infty(\Omega_0)$ in the set $D_0$ (with $b_{\max}\le 1.6\times10^6$) via $\tilde{\mathbb{D}}(\Omega_0)=\bigcup_{k=1}^{1500}\mathbb{C}^k(\Omega_0)$, where $\Omega_0$ is the set associated with $b_0= 100{,}000$, for which the inclusion $\Omega_0 \subseteq \inte_\X\big(\mathbb{C}(\Omega_0)\big)$ holds.
\textbf{A.} The triangular RCCS $\Omega_0$ and its approximated robust infinite controllable set $\mathbb{C}^{1500}(\Omega_0)$ are the sets enclosed in the black and blue contours, respectively. 
\textbf{B.} The RCCS $\Omega_0$ is shown along with $\mathbb{C}^{k}(\Omega_0)$ for $k=1,\ldots,1500$. 
\textbf{C.} The RCCS $\Omega_0$ is shown along with a zoom into the first controllable sets, confirming the expected inclusion $\mathbb{C}^{k-1}(\Omega)\subseteq \inte_\X\big(\mathbb{C}^k(\Omega)\big)$.}
	\label{fig:Domain}
\end{figure*}

The existence of a RCCS $\Omega_0$, shown in Fig.~\ref{fig:RCCS}, implies the existence of an AG-function and, consequently, the attractivity of the identified RCCS under the dynamics~\eqref{eq:Ex_system}.

We then aim to approximate $\mathbb{C}_\infty(\Omega_0)$, i.e., the set of initial conditions from which the bacterial population can be driven to $\Omega_0$ under antibiotic action and/or immune response.
To approximate $\mathbb{C}_\infty(\Omega_0)$ within $D_0$ with $b_{\max} \le 1.6\times10^6$, we select a set $\Omega_0$ with $b_0 \le 10,000$, for which the method ensures $\Omega_0 \subset \inte_\X\big(\mathbb{C}(\Omega_0)\big)$, and then estimate $\mathbb{C}_\infty(\Omega_0)$ by computing $\mathbb{C}^k(\Omega_0)$ for $k = 1,\ldots,1500$, thus approximating the infinite union through a finite union of sets.
Fig.~\ref{fig:Domain}-A shows the estimated approximation $\tilde{\mathbb{D}}(\Omega_0)=\mathbb{C}^{1500}(\Omega_0)$ within $D_0$; Fig.~\ref{fig:Domain}-B shows all sets $\mathbb{C}^k(\Omega_0)$ for $k=1,\ldots,1500$; and Fig.~\ref{fig:Domain}-C provides a closer view of the first controllable sets, highlighting that, as expected, $\mathbb{C}^{k-1}(\Omega_0)\subseteq \inte_\X\big(\mathbb{C}^k(\Omega_0)\big)$.

\section{Conclusion}

We presented a methodology to certify set-attractivity for uncertain nonlinear switched systems with state constraints, where the switching law is a decision variable. The approach is based on establishing the existence of an AG-function, which is a set-based map providing a decrease condition on the distance between the system trajectory and a prescribed target set in the state space.
Building upon controllability properties of contractive sets, we proposed a constructive procedure to obtain a particular AG-function. Importantly, the existence of a robust control contractive set guarantees the existence of an AG-function, and hence the robust local attractivity of the set itself.

The applicability of the methodology across different system classes was demonstrated by numerical simulations for linear examples and for a nonlinear case study arising in antimicrobial resistance dynamics.

Directions for future work include seeking milder and less conservative conditions to ensure the existence of an AG-function, beyond those derived from contractive sets, as well as performance-oriented approaches to approximate the region of attraction and to select, in the switching law $\sigma(x)$ in \eqref{eq:switchinglaw}, the optimal mode $\sigma^1$ (instead of the smallest) according to a desired performance-related functional.
Additionally, the proposed methodology can be further developed for the design of sequential drug therapies in antimicrobial resistance models.

\begin{ack}
We are grateful to Rami Katz for helpful suggestions, particularly concerning the proof of Proposition~\ref{propo:cont_nesetedC}.  We also thank Mirko Fiacchini for helpful discussions on the topics of this work.
\end{ack}

\bibliography{Reference}

@article{noghredani2021robust,
  title={Robust adaptive control for a class of nonlinear switched systems using state-dependent switching},
  author={Noghredani, Naeimadeen and Pariz, Naser},
  journal={SN Applied Sciences},
  volume={3},
  number={3},
  pages={290},
  year={2021},
  publisher={Springer}
}

@article{chen2018reachable,
  title={Reachable set estimation for switched positive systems},
  author={Chen, Yong and Lam, James and Shen, Jun and Du, Baozhu and Li, Panshuo},
  journal={International Journal of Systems Science},
  volume={49},
  number={11},
  pages={2341--2352},
  year={2018},
  publisher={Taylor \& Francis}
}

@article{decardi2021computing,
  title={Computing robust control invariant sets of constrained nonlinear systems: A graph algorithm approach},
  author={Decardi-Nelson, Benjamin and Liu, Jinfeng},
  journal={Computers \& Chemical Engineering},
  volume={145},
  pages={107177},
  year={2021},
  publisher={Elsevier}
}

@article{hajiahmadi2016robust,
  title={Robust $H_\infty$ switching control techniques for switched nonlinear systems with application to urban traffic control},
  author={Hajiahmadi, Mohammad and De Schutter, Bart and Hellendoorn, Hans},
  journal={International Journal of Robust and Nonlinear Control},
  volume={26},
  number={6},
  pages={1286--1306},
  year={2016},
  publisher={Wiley Online Library}
}

@inproceedings{aguiar2005stability,
  title={Stability of switched seesaw systems with application to the stabilization of underactuated vehicles},
  author={Aguiar, A Pedro and Hespanha, Joao P and Pascoal, Ant{\'o}nio M},
  booktitle={Proceedings of the 44th IEEE Conference on Decision and Control},
  pages={4584--4589},
  year={2005},
  organization={IEEE}
}

@article{shorten2007stability,
  title={Stability criteria for switched and hybrid systems},
  author={Shorten, Robert and Wirth, Fabian and Mason, Oliver and Wulff, Kai and King, Christopher},
  journal={SIAM review},
  volume={49},
  number={4},
  pages={545--592},
  year={2007},
  publisher={SIAM}
}

@article{baldi2018reachable,
  title={Reachable set estimation for switched linear systems with dwell-time switching},
  author={Baldi, Simone and Xiang, Weiming},
  journal={Nonlinear Analysis: Hybrid Systems},
  volume={29},
  pages={20--33},
  year={2018},
  publisher={Elsevier}
}

@article{sun2011stability,
  title={Stability theory of switched dynamical systems},
  author={Sun, Zhendong and Ge, Shuzhi Sam},
  year={2011},
  publisher={Springer}
}

@article{zhang2020stability,
  title={Stability analysis of discrete-time switched positive nonlinear systems with unstable subsystems under different switching strategies},
  author={Zhang, Niankun and Kang, Yu and Yu, Peilong},
  journal={IEEE Transactions on Circuits and Systems II: Express Briefs},
  volume={68},
  number={6},
  pages={1957--1961},
  year={2020},
  publisher={IEEE}
}

@article{chen2016estimation,
  title={Estimation and synthesis of reachable set for switched linear systems},
  author={Chen, Yong and Lam, James and Zhang, Baoyong},
  journal={Automatica},
  volume={63},
  pages={122--132},
  year={2016},
  publisher={Elsevier}
}

@article{zagabe2025uniform,
  title={Uniform global stability of switched nonlinear systems in the Koopman operator framework},
  author={Zagabe, Christian Mugisho and Mauroy, Alexandre},
  journal={SIAM Journal on Control and Optimization},
  volume={63},
  number={1},
  pages={472--501},
  year={2025},
  publisher={SIAM}
}

@article{yang2014survey,
  title={A survey of results and perspectives on stabilization of switched nonlinear systems with unstable modes},
  author={Yang, Hao and Jiang, Bin and Cocquempot, Vincent},
  journal={Nonlinear Analysis: Hybrid Systems},
  volume={13},
  pages={45--60},
  year={2014},
  publisher={Elsevier}
}

@article{niu2013robust,
  title={Robust stabilization and tracking control for a class of switched nonlinear systems},
  author={Niu, Ben and Zhao, Jun},
  journal={Asian Journal of Control},
  volume={15},
  number={5},
  pages={1496--1502},
  year={2013},
  publisher={Wiley Online Library}
}

@article{lee2008uniform,
  title={Uniform asymptotic stability of nonlinear switched systems with an application to mobile robots},
  author={Lee, Ti-Chung and Jiang, Zhong-Ping},
  journal={IEEE Transactions on Automatic Control},
  volume={53},
  number={5},
  pages={1235--1252},
  year={2008},
  publisher={IEEE}
}

@article{zhao2008stability,
  title={On stability, {$L_2$}-gain and {$H_\infty$} control for switched systems},
  author={Zhao, Jun and Hill, David J},
  journal={Automatica},
  volume={44},
  number={5},
  pages={1220--1232},
  year={2008},
  publisher={Elsevier}
}

@article{anderson2025invariant,
  title={Invariant set theory for predicting potential failure of antibiotic cycling},
  author={Anderson, Alejandro and Kinahan, Matthew W and Gonzalez, Alejandro H and Udekwu, Klas and Hernandez-Vargas, Esteban A},
  journal={Infectious Disease Modelling},
  volume={10},
  number={3},
  pages={897--908},
  year={2025},
  publisher={Elsevier}
}

@article{sun2013stability,
  title={On stability of a class of switched nonlinear systems},
  author={Sun, Yuangong and Wang, Long},
  journal={Automatica},
  volume={49},
  number={1},
  pages={305--307},
  year={2013},
  publisher={Elsevier}
}

@article{liu2015stability,
  title={Stability analysis of a class of nonlinear positive switched systems with delays},
  author={Liu, Xingwen},
  journal={Nonlinear Analysis: Hybrid Systems},
  volume={16},
  pages={1--12},
  year={2015},
  publisher={Elsevier}
}

@article{khoa2022exponential,
title = {Exponential stability analysis for a class of switched nonlinear time-varying functional differential systems},
journal = {Nonlinear Analysis: Hybrid Systems},
volume = {44},
pages = {101177},
year = {2022},
author = {Son {Nguyen Khoa} and Van Ngoc Le}
}

@article{anderson2021,
title = {Discrete-time MPC for switched systems with applications to biomedical problems},
journal = {Communications in Nonlinear Science and Numerical Simulation},
volume = {95},
pages = {105586},
year = {2021},
author = {A. Anderson and A.H. González and A. Ferramosca and E.A. Hernandez-Vargas}
}

@ARTICLE{hernandez2014,
  author={Hernandez-Vargas, Esteban Abelardo and Colaneri, Patrizio and Middleton, Richard H.},
  journal={IEEE Transactions on Control Systems Technology}, 
  title={Switching Strategies to Mitigate HIV Mutation}, 
  year={2014},
  volume={22},
  number={4},
  pages={1623-1628},
}

@article{hernandez2013,
title = {Optimal therapy scheduling for a simplified HIV infection model},
journal = {Automatica},
volume = {49},
number = {9},
pages = {2874-2880},
year = {2013},
author = {Esteban A. Hernandez-Vargas and Patrizio Colaneri and Richard H. Middleton}
}

@INPROCEEDINGS{hernandez2012,
  author={Hernandez-Vargas, Esteban A. and Colaneri, Patrizio and Middleton, Richard H.},
  booktitle={2012 IEEE 51st IEEE Conference on Decision and Control (CDC)}, 
  title={Sub-optimal switching with dwell time constraints for control of viral mutation}, 
  year={2012},
  pages={4906-4911},
}

@INPROCEEDINGS{devia2019,
  author={Devia, Carlos Andrés and Giordano, Giulia},
  booktitle={2019 IEEE 58th Conference on Decision and Control (CDC)}, 
  title={Optimal duration and planning of switching treatments taking drug toxicity into account: a convex optimisation approach}, 
  year={2019},
  pages={5674-5679},
}

@INPROCEEDINGS{giordano2016,
  author={Giordano, Giulia and Rantzer, Anders and Jonsson, Vanessa D.},
  booktitle={2016 IEEE 55th Conference on Decision and Control (CDC)}, 
  title={A convex optimization approach to cancer treatment to address tumor heterogeneity and imperfect drug penetration in physiological compartments}, 
  year={2016},
  pages={2494-2500},
}

@article{gandra2014economic,
  title={Economic burden of antibiotic resistance: how much do we really know?},
  author={Gandra, S and Barter, DM and Laxminarayan, R},
  journal={Clinical microbiology and infection},
  volume={20},
  number={10},
  pages={973--980},
  year={2014},
  publisher={Elsevier}
}

@article{serry2023zonotopic,
  title={Zonotopic Under-Approximations of Input Reachable Sets for Controllable Linear Systems},
  author={Serry, Mohamed and Liu, Jun},
  journal={IEEE Control Systems Letters},
  volume={7},
  pages={1453--1458},
  year={2023},
  publisher={IEEE}
}

@inproceedings{kerrigan2000invariant,
  title={Invariant sets for constrained nonlinear discrete-time systems with application to feasibility in model predictive control},
  author={Kerrigan, Eric C and Maciejowski, Jan M},
  booktitle={Proceedings of the 39th IEEE conference on decision and control (Cat. No. 00CH37187)},
  volume={5},
  pages={4951--4956},
  year={2000},
  organization={IEEE}
}

@article{lin2007switching,
  title={Switching stabilizability for continuous-time uncertain switched linear systems},
  author={Lin, Hai and Antsaklis, Panos J},
  journal={IEEE Transactions on automatic control},
  volume={52},
  number={4},
  pages={633--646},
  year={2007},
  publisher={IEEE}
}

@article{son2020robust,
  title={On robust stability of switched linear systems},
  author={Son, Nguyen Khoa and Van Ngoc, Le},
  journal={IET Control Theory \& Applications},
  volume={14},
  number={1},
  pages={19--29},
  year={2020},
  publisher={Wiley Online Library}
}

@article{allerhand2010robust,
  title={Robust stability and stabilization of linear switched systems with dwell time},
  author={Allerhand, Liron I and Shaked, Uri},
  journal={IEEE Transactions on Automatic Control},
  volume={56},
  number={2},
  pages={381--386},
  year={2010},
  publisher={IEEE}
}

@article{anderson2025computational,
  title={Computational framework for streamlining the success of sequential antibiotic therapy},
  author={Anderson, Alejandro and Kinahan, Matthew W and Blanco-Rodriguez, Rodolfo and Gonzalez, Alejandro H and Udekwu, Klas and Hernandez-Vargas, Esteban A},
  journal={npj Antimicrobials and Resistance},
  volume={3},
  number={1},
  pages={90},
  year={2025},
  publisher={Nature Publishing Group UK London}
}

@inproceedings{cavallo2025insulin,
  title={Insulin Sensitivity Management in Artificial Pancreas: a Switching Control Strategy Approach--An In Silico Study},
  author={Cavallo, Maria Sofia and Licini, Nicola and Previdi, Fabio and Ferramosca, Antonio},
  booktitle={2025 IEEE 64th Conference on Decision and Control (CDC)},
  pages={5756--5761},
  year={2025},
  organization={IEEE}
}

@article{wu2022switched,
  title={Switched system optimal control approach for drug administration in cancer chemotherapy},
  author={Wu, Xiang and Hou, Yuzhou and Zhang, Kanjian},
  journal={Biomedical Signal Processing and Control},
  volume={75},
  pages={103575},
  year={2022},
  publisher={Elsevier}
}

@article{deaecto2025stabilisation,
  title={Stabilisation of discrete-time switched nonlinear systems},
  author={Deaecto, Grace S and Astolfi, Alessandro},
  journal={IEEE Transactions on Automatic Control},
  year={2025},
  publisher={IEEE}
}

@article{colaneri2008stabilization,
  title={Stabilization of continuous-time switched nonlinear systems},
  author={Colaneri, Patrizio and Geromel, Jos{\'e} C and Astolfi, Alessandro},
  journal={Systems \& Control Letters},
  volume={57},
  number={1},
  pages={95--103},
  year={2008},
  publisher={Elsevier}
}

@article{heemels2016lyapunov,
  title={On Lyapunov-Metzler inequalities and S-procedure characterizations for the stabilization of switched linear systems},
  author={Heemels, WP Maurice H and Kundu, Atreyee and Daafouz, Jamal},
  journal={IEEE Transactions on Automatic Control},
  volume={62},
  number={9},
  pages={4593--4597},
  year={2016},
  publisher={IEEE}
}

@article{jungers2017feedback,
  title={On feedback stabilization of linear switched systems via switching signal control},
  author={Jungers, Rapha{\"e}l M and Mason, Paolo},
  journal={SIAM Journal on Control and Optimization},
  volume={55},
  number={2},
  pages={1179--1198},
  year={2017},
  publisher={SIAM}
}

@article{deaecto2018stability,
  title={Stability and performance of discrete-time switched linear systems},
  author={Deaecto, Grace S and Geromel, Jose C},
  journal={Systems \& Control Letters},
  volume={118},
  pages={1--7},
  year={2018},
  publisher={Elsevier}
}

@article{geromel2006stability,
  title={Stability and stabilization of discrete time switched systems},
  author={Geromel, Jos{\'e} C and Colaneri, Patrizio},
  journal={International Journal of Control},
  volume={79},
  number={07},
  pages={719--728},
  year={2006},
  publisher={Taylor \& Francis}
}

@inproceedings{russo2022state,
  title={State dependent switching control of affine linear systems with dwell time: Application to power converters},
  author={Russo, Antonio and Incremona, Gian Paolo and Cavallo, Alberto and Colaneri, Patrizio},
  booktitle={2022 American Control Conference (ACC)},
  pages={3807--3813},
  year={2022},
  organization={IEEE}
}

@incollection{doan2013hierarchical,
  title={A hierarchical {MPC} approach with guaranteed feasibility for dynamically coupled linear systems},
  author={Doan, MD and Keviczky, T and Schutter, B De},
  booktitle={Distributed Model Predictive Control Made Easy},
  pages={393--406},
  year={2013},
  publisher={Springer}
}

@article{fiacchini2021yet,
  title={Yet another computation-oriented necessary and sufficient condition for stabilizability of switched linear systems},
  author={Fiacchini, Mirko},
  journal={IEEE Transactions on Automatic Control},
  volume={67},
  number={7},
  pages={3627--3632},
  year={2021},
  publisher={IEEE}
}

@article{cinto2025switching,
  title={Switching-Controlled Invariance of Polytopic Sets for Switched Affine Systems with Dwell Time},
  author={Cinto, Felipe and Vallarella, Alexis J and Russo, Antonio and Haimovich, Hernan},
  journal={SIAM Journal on Control and Optimization},
  volume={63},
  number={5},
  pages={3167--3188},
  year={2025},
  publisher={SIAM}
}

@inproceedings{xiang2017reachable,
  title={On reachable set estimation for discrete-time switched linear systems under arbitrary switching},
  author={Xiang, Weiming and Tran, Hoang-Dung and Johnson, Taylor T},
  booktitle={American Control Conference (ACC)},
  pages={4534--4539},
  year={2017},
}

@article{liu2019separable,
  title={Separable {L}yapunov-like functions for switched positive non-linear systems via a contractive approach},
  author={Liu, Qian},
  journal={IET Control Theory \& Applications},
  volume={13},
  number={7},
  pages={943--951},
  year={2019},
  publisher={Wiley Online Library}
}

@article{perez2025characterization,
  title={Characterization and Computation of a Control Invariant set for Discrete-Time Switched Systems Under Waiting-Time Constraints},
  author={Perez, Mara and Sanchez, Ignacio and Anderson, Alejandro and Gonz{\'a}lez, Alejandro H and Actis, Marcelo},
  journal={IEEE Control Systems Letters},
  year={2025},
  publisher={IEEE}
}

@article{danielson2019necessary,
  title={Necessary and sufficient conditions for constraint satisfaction in switched systems using switch-robust control invariant sets},
  author={Danielson, Claus and Bridgeman, Leila J and Di Cairano, Stefano},
  journal={International Journal of Robust and Nonlinear Control},
  volume={29},
  number={9},
  pages={2589--2602},
  year={2019},
  publisher={Wiley Online Library}
}

@inproceedings{anderson2024stabilizability,
  title={Stabilizability of uncertain switched systems to characterize antibiotic resistance evolution},
  author={Anderson, A and Ohemeng, MO and Gonzalez, AH and Hernandez-Vargas, EA},
  booktitle={IEEE Conference on Decision and Control (CDC)},
  pages={7050--7055},
  year={2024},
}

@book{hairer1993solving,
  title={Solving ordinary differential equations I: Nonstiff problems},
  author={Hairer, Ernst and Wanner, Gerhard and N{\o}rsett, Syvert P},
  year={1993},
  publisher={Springer}
}

@article{kowalska2021minimum,
  title={The minimum inhibitory concentration of antibiotics: methods, interpretation, clinical relevance},
  author={Kowalska-Krochmal, Beata and Dudek-Wicher, Ruth},
  journal={Pathogens},
  volume={10},
  number={2},
  pages={165},
  year={2021},
  publisher={MDPI}
}

@article{anderson2025failure,
  title={Failure and success in single-drug control of antimicrobial resistance},
  author={Anderson, Alejandro and Katz, Rami and Cal{\`a} Campana, Francesca and Giordano, Giulia},
  journal={IEEE Control Systems Letters},
volume={9},
pages={991-996},
  year={2025},
  publisher={IEEE}
}

@inproceedings{herceg2013multi,
  title={Multi-parametric toolbox 3.0},
  author={Herceg, Martin and Kvasnica, Michal and Jones, Colin N and Morari, Manfred},
  booktitle={European control conference (ECC)},
  pages={502--510},
  year={2013},
}

@article{blanchini1994ultimate,
  title={Ultimate boundedness control for uncertain discrete-time systems via set-induced {L}yapunov functions},
  author={Blanchini, Franco},
  journal={IEEE Transactions on Automatic Control},
  volume={39},
  number={2},
  pages={428--433},
  year={1994},
  publisher={IEEE}
}

@article{mason2023universal,
  title={On universal classes of {L}yapunov functions for linear switched systems},
  author={Mason, Paolo and Chitour, Yacine and Sigalotti, Mario},
  journal={Automatica},
  volume={155},
  pages={111155},
  year={2023},
  publisher={Elsevier}
}

@article{Tetteh2023,
author = {Tetteh, Josephine N. A. and Olaru, Sorin and Crauel, Hans and Hernandez-Vargas, Esteban A.},
title = {Scheduling collateral sensitivity-based cycling therapies toward eradication of drug-resistant infections},
journal = {International Journal of Robust and Nonlinear Control},
volume = {33},
number = {9},
pages = {4824-4842},
year = {2023}
}

@inproceedings{hernandez2021switching,
  title={Switching Logistic Maps to Design Cycling Approaches Against Antimicrobial Resistance},
  author={Hernandez-Vargas, Esteban A and Parra-Rojas, C{\'e}sar and Olaru, Sorin},
  booktitle={2021 60th IEEE Conference on Decision and Control (CDC)},
  pages={4248--4253},
  year={2021},
  organization={IEEE}
}

@article{kouyos2011informed,
  title={Informed switching strongly decreases the prevalence of antibiotic resistance in hospital wards},
  author={Kouyos, Roger D and Abel zur Wiesch, Pia and Bonhoeffer, Sebastian},
  journal={PLoS computational biology},
  volume={7},
  number={3},
  pages={e1001094},
  year={2011},
  publisher={Public Library of Science San Francisco, USA}
}

@article{lin2009stability,
  title={Stability and stabilizability of switched linear systems: a survey of recent results},
  author={Lin, Hai and Antsaklis, Panos J},
  journal={IEEE Transactions on Automatic Control},
  volume={54},
  number={2},
  pages={308--322},
  year={2009},
  publisher={IEEE}
}

@article{chen2019switched,
  title={A switched systems approach to path following with intermittent state feedback},
  author={Chen, Hsi-Yuan and Bell, Zachary I and Deptula, Patryk and Dixon, Warren E},
  journal={IEEE Transactions on Robotics},
  volume={35},
  number={3},
  pages={725--733},
  year={2019},
  publisher={IEEE}
}

@article{wang2019stability,
  title={Stability and stabilization of continuous-time switched systems: a multiple discontinuous convex Lyapunov function approach},
  author={Wang, Ruihua and Hou, Linlin and Zong, Guangdeng and Fei, Shumin and Yang, Dong},
  journal={International Journal of Robust and Nonlinear Control},
  volume={29},
  number={5},
  pages={1499--1514},
  year={2019},
  publisher={Wiley Online Library}
}

@article{gomide2018stability,
  title={Stability analysis of discrete-time switched systems under arbitrary switching},
  author={Gomide, Thales S and Lacerda, M{\'a}rcio J},
  journal={IFAC-PapersOnLine},
  volume={51},
  number={25},
  pages={371--376},
  year={2018},
  publisher={Elsevier}
}

@article{cousin2021switched,
  title={A switched Lyapunov-passivity approach to motorized FES cycling using adaptive admittance control},
  author={Cousin, Christian A and Deptula, Patryk and Rouse, Courtney A and Dixon, Warren E},
  journal={IEEE Transactions on Control Systems Technology},
  volume={30},
  number={2},
  pages={740--754},
  year={2021},
  publisher={IEEE}
}

@article{fiacchini2018stabilization,
  title={Stabilization and control {L}yapunov functions for language constrained discrete-time switched linear systems},
  author={Fiacchini, Mirko and Jungers, Marc and Girard, Antoine},
  journal={Automatica},
  volume={93},
  pages={64--74},
  year={2018},
  publisher={Elsevier}
}

@article{fiacchini2014necessary,
  title={Necessary and sufficient condition for stabilizability of discrete-time linear switched systems: A set-theory approach},
  author={Fiacchini, Mirko and Jungers, Marc},
  journal={Automatica},
  volume={50},
  number={1},
  pages={75--83},
  year={2014},
  publisher={Elsevier}
}

@book{liberzon2003switching,
  title={Switching in systems and control},
  author={Liberzon, Daniel},
  volume={190},
  year={2003},
  publisher={Springer}
}

@article{liberzon1999basic,
  title={Basic problems in stability and design of switched systems},
  author={Liberzon, Daniel and Morse, A Stephen},
  journal={IEEE Control Systems Magazine}, 
  title={Basic problems in stability and design of switched systems}, 
  year={1999},
  volume={19},
  number={5},
  pages={59--70},
  }

@book{Blanchinibook15,
  title={Set-Theoretic Methods in Control. Systems \& Control: Foundations \& Applications},
  author={Blanchini, Franco and Miani, Stefano},
  year={2015},
  publisher={Springer}
}

@article{hernandez2011discrete,
  title={Discrete-time control for switched positive systems with application to mitigating viral escape},
  author={Hernandez-Vargas, Esteban and Colaneri, Patrizio and Middleton, Richard and Blanchini, Franco},
  journal={International journal of robust and nonlinear control},
  volume={21},
  number={10},
  pages={1093--1111},
  year={2011},
  publisher={Wiley Online Library}
}

@article{anderson2023cyclic,
  title={Cyclic control equilibria for switched systems with applications to ecological systems},
  author={Anderson, Alejandro L and Abuin, Pablo and Ferramosca, Antonio and Hernandez-Vargas, Esteban A and Gonzalez, Alejandro H},
  journal={International Journal of Robust and Nonlinear Control},
  volume={33},
  number={9},
  pages={5159--5175},
  year={2023},
  publisher={Wiley Online Library}
}

@article{geromel2006stabilityCT,
author = {Geromel, J. C. and  Colaneri, Patrizio},
title = {Stability and stabilization of continuous--time switched linear systems},
journal = {SIAM Journal on Control and Optimization},
volume = {45},
number = {5},
pages = {1915--1930},
year = {2006},
}

@ARTICLE{hespanha2004,
  author={Hespanha, J.P.},
  journal={IEEE Transactions on Automatic Control}, 
  title={Uniform stability of switched linear systems: extensions of {LaSalle}'s {I}nvariance {P}rinciple}, 
  year={2004},
  volume={49},
  number={4},
  pages={470-482},
 }

@ARTICLE{egidio2019,
  author={Egidio, L. N. and Deaecto, G. S.},
  journal={IEEE Trans. on Automatic Control}, 
  title={Novel Practical Stability Conditions for Discrete-Time Switched Affine Systems}, 
  year={2019},
  volume={64},
  number={11},
  pages={4705-4710},
}

\end{document}